\documentclass[11pt, onecolumn]{article}
\usepackage[top=1in, bottom=1in, left=1.25in, right=1.25in]{geometry}

\usepackage{amsfonts}
\usepackage{amsmath,amssymb}
\usepackage{graphicx}
\usepackage{color, soul}
\usepackage{algorithm,algorithmic}
\usepackage{bm}
\usepackage{booktabs}
\usepackage{flushend}
\usepackage{tikz}
\usetikzlibrary{arrows}
\usepackage{subcaption}

\usepackage[amsmath,thmmarks]{ntheorem}
\usepackage{theorem}

\newtheorem{cor}{Corollary}
\newtheorem{lem}{Lemma}
\newtheorem{defi}{Definition}
\newtheorem{rem}{Remark}
\newtheorem{thm}{Theorem}

\theoremheaderfont{\sc}\theorembodyfont{\upshape}
\theoremstyle{nonumberplain}
\theoremseparator{}
\theoremsymbol{\rule{1ex}{1ex}}
\newtheorem{proof}{Proof}

\renewcommand{\arraystretch}{1.5}

\newcommand{\mythbox}[3]{\parbox[c]{#1}{\vspace{#2}\centering #3 \vspace{#2}}}

\newcommand{\aff}{{\rm aff}}
\newcommand{\oaff}{\overline{{\rm aff}}}
\newcommand{\set}{\mathcal}
\newcommand{\od}{\overline{{D}}}
\newcommand{\affX}{\aff_{\set X}}
\newcommand{\affXS}{\aff_{\set X}^2}
\newcommand{\affY}{\aff_{\set Y}}
\newcommand{\affYS}{\aff_{\set Y}^2}
\newcommand{\oaffYS}{\oaff_{\set Y}^2}

\newcommand{\dX}{D_{\set X}}
\newcommand{\dXS}{D_{\set X}^2}
\newcommand{\dY}{D_{\set Y}}
\newcommand{\dYS}{D_{\set Y}^2}
\newcommand{\odYS}{\od_{\set Y}^2}

\newcommand{\va}{{\bf a}}

\newcommand{\Proj}{{\rm P}}
\newcommand{\Span}{{\set C}}

\begin{document}

\title{Rigorous Restricted Isometry Property of Low-Dimensional Subspaces}

\author{Gen Li, Qinghua Liu, and Yuantao Gu%
\thanks{The authors are with the Department of Electronic Engineering, Tsinghua University, Beijing 100084, China. 
The corresponding author of this paper is Yuantao Gu (gyt@tsinghua.edu.cn).}}

\date{submitted January 30, 2018}

\maketitle

%
%

\begin{abstract}
	Dimensionality reduction is in demand to reduce the complexity of solving large-scale problems 
	with data lying in latent low-dimensional structures
	in machine learning and computer version.
	Motivated by such need,
	in this work we study the Restricted Isometry Property (RIP) of Gaussian random projections for 
	low-dimensional subspaces in $\mathbb{R}^N$, and rigorously prove that
	the projection Frobenius norm distance between any two subspaces spanned by the projected data in $\mathbb{R}^n$
	($n<N$)
	remain almost the same as the distance between the original subspaces
	with probability no less than $1 - {\rm e}^{-\mathcal{O}(n)}$.
	Previously the well-known Johnson-Lindenstrauss (JL) Lemma and RIP for sparse vectors 
	have been the foundation of sparse signal processing including Compressed Sensing.
	As an analogy to JL Lemma and RIP for sparse vectors, this work allows the use of random projections 
	to reduce the ambient dimension with the theoretical guarantee that the distance between subspaces after compression
	is well preserved.
	
\end{abstract}

\section{Introduction}

This paper studies the Restricted Isometry Property (RIP) of random projections for subspaces.
It reveals that the distance between two low-dimensional subspaces remain almost unchanged
after being projected by a Gaussian random matrix with overwhelming probability, 
when the ambient dimension after projection is sufficiently large in comparison with the dimension of subspaces.

\subsection{Motivation}

In the era of data deluge, labeling huge amount of large-scale data can be time-consuming, costly, and even intractable,
so unsupervised learning has attracted increasing attention in recent years. 
One of such methods emerging recently, subspace clustering (SC) \cite{Elhamifar2009Sparse, soltanolkotabi2012geometric, elhamifar2013sparse, Heckel2015Robust}, which depicts the latent structure of a variety of data as a union of subspaces, 
has been shown to be powerful in a wide range of applications, including motion segmentation, face clustering, and anomaly detection. 
It also shows great potential to some previously less explored datasets, such as network data, gene series, and medical images.

Traditional subspace clustering methods, however, suffer from the deficiency in similarity representation,
so it can be computationally expensive to adapt them to large-scale datasets.
In order to alleviate the high computational burden, 
a variety of works have been done to address the crucial problem of
how to efficiently handle large-scale datasets.
Compressed Subspace Clustering (CSC) \cite{Mao2014Compressed} 
also known as Dimensionality-reduced Subspace Clustering \cite{Heckel2014Subspace}
is a method that performs SC on randomly compressed data points.
Because the random compression reduces the dimension of the ambient space, 
the computational cost of finding the self-representation in SC can be efficiently reduced. 
Based on the concept of subspace affinity, which characterizes the similarity between two subspaces, and the mathematical tools introduced in \cite{soltanolkotabi2012geometric}, the conditions under which several popular algorithms can successfully 
cluster the compressed data have been theoretically studied and numerically verified \cite{heckel2015dimensionality,Wang2016}.

Because the data points are randomly projected from a high-dimensional ambient space $\mathbb{R}^N$
to a new medium-dimensional ambient space $\mathbb{R}^n$, 
a worry is that the similarity between any two low-dimensional subspaces increases
and the SC algorithms are less likely to perform well.
Inspired by the well-known Johnson-Lindenstrauss (JL) Lemma \cite{Johnson1984Extensions, dasgupta1999elementary} and the Restricted Isometry Property (RIP) \cite{Candes2005Decoding, Cand2008The, Baraniuk2015A}, which allows the use of random projection to reduce the space dimension while keeping the Euclidean distance between any two data points 
and leads to the boom of sparse signal processing including 
Compressed Sensing (CS) \cite{Donoho2006Compressed, Candes2006Robust, Aeron2010Information, candes2007sparsity, eldar2012compressed, Eftekhari2015New, Kutyniok2009Robust},
one may speculate whether the similarity (or distance) between any two given subspaces can remain almost unchanged,
if the dimension of the latent subspace that the data lie in is small compared with that of the ambient space after projection $n$.
It should be highlighted that this conjecture is not confined to the SC problem,
so we believe that it may benefit future studies on other subspace related topics.

Motivated by the conjecture about whether the similarity between any two given subspaces can remain almost unchanged after random projection, we study the RIP of Gaussian random projections for a finite set of subspaces.
In order to give more solid guarantees and more precise insight into the law of magnitude of the dimensions
for CSC and other subspace related problems,
we derive an optimum probability bound of the RIP of Gaussian random compressions for subspaces in this paper.
Compared with our previous work \cite{li2017restricted}, the probability bound has been improve from
$1 - \mathcal{O}(1/n)$ to $1 - {\rm e}^{-\mathcal{O}(n)}$, 
which is optimum when we consider the state-of-the-art statistical probability theories for Gaussian random matrix.

\subsection{Main Results}

The projection Frobenius norm (F-norm for short) distance is adopted in this work to measure the distance between two subspaces.
It should be noted that we slightly generalize the definition in \cite{SeveralDistances1998} 
to the situation where the dimensions of the two subspaces are different.

\begin{defi}[\cite{li2017restricted}]{\bf (Projection Frobenius norm distance between subspaces)}
The generalized projection F-norm distance between two subspaces ${\set X}_1$ and ${\set X}_2$ is defined as
$$
D({\set X}_1,{\set X}_2) :=\frac1{\sqrt 2} \|{\bf U}_1{\bf U}_1^{\rm T} - {\bf U}_2{\bf U}_2^{\rm T}\|_{\rm F},
$$
where ${\bf U}_i$ denotes an arbitrary orthonormal basis matrix for subspace ${\set X}_i, i=1,2$.
\end{defi}

We will focus on the change of the distance between any 
 two low-dimensional subspaces after being randomly projected
 from $\mathbb{R}^N$ to $\mathbb{R}^n$ ($n<N$).
The projection of a low-dimensional subspace by using a Gaussian random matrix is defined as below.

\begin{defi}\label{defi-project-subspace}{\bf (Gaussian random projection for subspace)}
The Gaussian random projection of a $d$-dimensional subspace $\set{X}\subset \mathbb{R}^N$ onto $\mathbb{R}^n$ ($d < n<N$) is defined as below,
$$
	\set{X} \stackrel{\bm \Phi}{\longrightarrow} \set{Y} = \{{\bf y} | {\bf y}={\bm \Phi}{\bf x}, \forall {\bf x}\in \set{X}\},
$$
where the projection matrix ${\bm \Phi}\in\mathbb{R}^{n\times N}$ is composed of entries independently drawn from Gaussian distribution $\mathcal{N}(0,1/n)$.
\end{defi}

One may notice that the dimensions of subspaces remain unchanged after random projection with probability one.

Based on the definitions above, the main theoretical result of this work is stated as follows.

\begin{thm}\label{thm-rip}
Suppose $\set X_1, \ldots, \set X_L\subset \mathbb{R}^N$ are $L$ subspaces with dimension less than $d$.
After random projection by using a Gaussian random matrix $\bm\Phi\in\mathbb{R}^{n\times N}$, $\set{X}_i \stackrel{\bm \Phi}{\longrightarrow} \set{Y}_i\subset\mathbb{R}^n, i=1,\cdots,L, n<N$.
There exist constants $c_1(\varepsilon),c_2(\varepsilon) > 0$ depending only on $\varepsilon$ such that for any two subspaces $\set X_i$ and $\set X_j$, for any $n > c_1(\varepsilon)\max\{d, \ln L\}$,
\begin{equation}\label{thm-sub-aff-bound}
    \left(1 - \varepsilon\right)D^2({\set X}_i,{\set X}_j) < D^2({\set Y}_i,{\set Y}_j) < \left(1 + \varepsilon \right)D^2({\set X}_i,{\set X}_j)
\end{equation}
holds with probability at least $1 - {\rm e}^{-c_2(\varepsilon)n}$.
\end{thm}

Theorem \ref{thm-rip} reveals that the distance between two subspaces remains almost unchanged after random projection with overwhelming probability, when the ambient dimension after projection $n$ is sufficiently large.

\subsection{Our Contribution}

In this paper, we study the RIP of Gaussian random matrices for projecting a finite set of subspaces.
The problem is challenging as random projections neither preserve orthogonality nor normalize the vectors defining orthonormal bases of the subspaces.
In order to measure the change in subspace distance induced by random projections, both effects have to be carefully quantified.
Based on building a metric space of subspaces with the projection F-norm distance,
which is closely connected with subspace affinity,
we start from verifying that the affinity between two subspaces concentrates on its estimate with overwhelming probability after Gaussian random projection.
Then we successfully reach the RIP of two subspaces and generalize it to the situation of a finite set of subspaces, as stated in Theorem \ref{thm-rip}.

The main contribution of this work is to provide a mathematical tool,
which can shed light on many problems including CSC.
As a direct result of Theorem \ref{thm-rip},
when solving the SC problem at a large scale,
one may conduct SC on randomly compressed samples 
to alleviate the high computational burden and still have theoretical performance guarantee.
Because the distance between subspaces almost remains unchanged after projection,
the clustering error rate of any SC algorithm may keep as small as that conducting in the original space.
Considering that our theory is independent of SC algorithms, this may benefit future studies on other subspace related topics.

Except our previous work \cite{li2017restricted} that will be compared with in Section \ref{Relatedworks}, 
as far as we know, there is no relevant work that study the distance preserving property between subspaces after random projection.

\subsubsection{Comparison with JL Lemma and RIP for Sparse Signals}

The famous Johnson-Lindenstrauss Lemma illustrates that there exists a map from a higher-dimensional space into a lower-dimensional space such that the distance between a finite set of data points will change little after being mapped.

\begin{lem}[JL Lemma]\cite{Johnson1984Extensions, dasgupta1999elementary}
For any set ${\set V}$ of $L$ points in $\mathbb{R}^N$, there exists a map $f : \mathbb{R}^N \to \mathbb{R}^n, n < N$, such that for all ${\bf x}_1, {\bf x}_2 \in {\set V}$,
$$
(1-\varepsilon)\|{\bf x}_1 - {\bf x}_2\|_2^2 \le \|f({\bf x}_1) - f({\bf x}_2)\|_2^2 \le (1+\varepsilon)\|{\bf x}_1 - {\bf x}_2\|_2^2
$$
if $n$ is a positive integer satisfying
$
n \ge {4{\rm ln} L}/({{\varepsilon^2}/{2} - {\varepsilon^3}/{3}}),
$
where $0 < \varepsilon < 1$ is a constant.
\end{lem}

The RIP of random matrix illustrates that the distance between two sparse vectors will change little with high probability after random projection.

\begin{defi} \cite{Candes2005Decoding, Cand2008The, Baraniuk2015A} \label{def}
The projection matrix ${\bm \Phi}\in\mathbb{R}^{n\times N}, n < N$ satisfies RIP of order $k$ if there exists a $\delta_k\in\left(0,1\right)$ such that
$$
(1\!-\!\delta_k)\|{\bf x}_1 \!-\! {\bf x}_2\|_2^2 \le \|{\bm \Phi}{\bf x}_1 \!-\! {\bm \Phi}{\bf x}_2\|_2^2 \le (1\!+\!\delta_k)\|{\bf x}_1 \!-\! {\bf x}_2\|_2^2
$$
holds for any two $k$-sparse vectors ${\bf x}_1, {\bf x}_2\in\mathbb{R}^N$.
\end{defi}

\begin{thm}\label{gauss-rand}\cite{Baraniuk2015A}
A Gaussian random matrix ${\bm \Phi}\in\mathbb{R}^{n\times N}, n < N$ has the RIP of order $k$ for
$
n \ge c_1k{\rm ln}\left(\frac{N}k\right)
$
with probability
$
1 - {\rm e}^{-c_2n},
$
where $c_1,c_2 > 0$ are constants depending only on $\delta_k$, the smallest nonnegative constant satisfying Definition \ref{def}.
\end{thm}

In summary, the above works focus on the change of the distance between points after determinate mapping or random projection. In comparison, our work views a subspace as a whole and studies the distance between subspaces, 
which to the best of our knowledge has never been studied before. 
Moreover, the above works study the points in Euclidean space with $l_2$-norm, 
while our work study the subspaces on the Grassmannian manifold with F-norm metric, 
which is highly nonlinear and more complex.
A detailed comparison to explain the differences between our work and related works is presented in Table \ref{fig}.

\begin{table}[!t]
\renewcommand{\arraystretch}{1.5}
\newcommand{\firstcol}[1]{\mythbox{5.1em}{1em}{#1}}
\newcommand{\secondcol}[1]{\mythbox{6em}{1em}{#1}}
\newcommand{\thirdcol}[1]{\mythbox{7em}{1em}{#1}}
\newcommand{\fourthcol}[1]{\mythbox{6em}{1em}{#1}}
\newcommand{\fifthcol}[1]{\mythbox{7em}{1em}{#1}}
\newcommand{\secondthirdcol}[1]{\mythbox{13em}{1em}{#1}}
\newcommand{\fourthfifthcol}[1]{\mythbox{13em}{1em}{#1}}
\caption{Comparison with other dimension-reduction theories including JL Lemma, RIP for sparse signals, and our previous results \cite{li2017restricted}}\label{fig}
\begin{center}
\begin{tabular}{c|c|c|c|c}
\toprule[1pt]
&  &  &
\multicolumn{2}{c}{\fourthfifthcol{RIP for \\low-dimensional subspaces}}
\\
\cline{4-5}
 &
\raisebox{4ex}[0pt]{JL Lemma} &
\raisebox{4ex}[0pt]{\thirdcol{RIP for \\sparse signals}} &
\fourthcol{\cite{li2017restricted}} &
\fifthcol{this work} \\
 \toprule
object & \secondcol{any set of $L$ points in $\mathbb{R}^N$} &
\thirdcol{all $k$-sparse signals in $\mathbb{R}^N$} &
\multicolumn{2}{c}{\fourthfifthcol{any set of $L$ $d$-dimensional \\ subspaces in $\mathbb{R}^N$}} \\ \hline
metric &
\multicolumn{2}{c|}{\secondthirdcol{Euclidean distance\\$\|{\bf x}_i - {\bf x}_j\|_2$}} &
\multicolumn{2}{c}{\fourthfifthcol{projection F-norm distance \\$\frac1{\sqrt 2} \|{\bf P}_i - {\bf P}_j\|_{\rm F}$}} \\ \hline
\firstcol{compression method} &
some map $f$ & \multicolumn{3}{c}{Gaussian random matrix}\\ \hline
error bound & $(1-\varepsilon, 1+\varepsilon)$ & $(1-\delta_k,1+\delta_k)$ & \multicolumn{2}{c}{$\left(1-\varepsilon,1+\varepsilon\right)$} \\ \hline
condition & $n \ge \frac{4{\rm ln} L}{{\varepsilon^2}/{2} - {\varepsilon^3}/{3}}$ & $n \ge c_1k{\rm ln}\left(\frac{N}k\right)$ & $n$ large enough & $n \!>\! c_1\!\max\{d, \ln L\}$ \\ \hline
\firstcol{success probability} & 1 & $1 - {\rm e}^{-c_2n}$ & 1$-\frac{2dL\left(L-1\right)}{\left(\varepsilon-d/n\right)^2n}$ & $1 - {\rm e}^{-c_2n}$ \\ \bottomrule[1pt]
\end{tabular}
\end{center}
\label{default}
\end{table}%

\subsubsection{Comparison with RIP for Signals in UoS}

There are literatures studying the distance preserving properties of compressed data points,
which may be sparse on specific basis or lie in a couple of subspaces or surfaces \cite{Davenport2010Signal, Blumensath2009Sampling, agarwal2007embeddings, magen2002dimensionality}.

The authors of \cite{Davenport2010Signal} extended the RIP to signals that are sparse or compressible with respect to a certain basis ${\bm \Psi}$, i.e., ${\bf x} = {\bm \Psi}{\bm \alpha}$, where ${\bm \Psi}$ is represented as a unitary $N \times N$ matrix and ${\bm \alpha}$ is a $k$-sparse vector.
The work of \cite{Blumensath2009Sampling} proves that with high probability the random projection matrix ${\bm \Phi}$ can preserve the distance between two signals belonging to a Union of Subspaces (UoS).
In \cite{ agarwal2007embeddings}, it is shown that random projection preserves the structure of surfaces.
Given a collection of $L$ surfaces of linearization dimension $d$, if they are embedded into 
a space of $\mathcal{O}(d\delta^{2} \log(Ld/\delta))$ dimension, the surfaces are preserved  
in the sense that for any pair of points on these surfaces the distance between them are preserved.
The main contribution of \cite{magen2002dimensionality} is stated as follows. If $S$ is an $n$ point subset of $\mathbb{R}^N$, $0 < \delta < \frac{1}{3}$ and
$
n = 256d\log n(\max\{d, 1/\delta\})^2,
$
there is a mapping of $\mathbb{R}^N$ into $\mathbb{R}^n$ under which volumes of sets of size at most $d$ do not change by more than a factor of $1+\delta$, and the distance of points from affine hulls of sets of size at most $k - 1$ is preserved within a relative error of $\delta$.

According to above survey, those works study embedding of Euclidean distances between points in subspaces, while we discuss embedding of a finite set of subspaces in terms of the projection F-norm distance.
In both the related works and this paper, the same mathematical tool of concentration inequalities and random matrix theory are adopted to derive the RIP for two different objects, i.e., data points in Euclidean space and subspaces in Euclidean space (or points on Grassmann manifold), respectively.
In comparison, both Euclidean space and random projection are linear,
but Grassmannian is not linear, let along the projection on it,
so the new problem is much more difficult than the existing one,
and a core contribution of this work is dealing with the above challenges with a brand-new geometric proof,
the technique in which has hardly been used previously to derive the RIP for data points.

\subsection{Organization}

The rest of this paper is organized as follows.
Based on the introduction of principal angles, affinity, and its connection with the projection F-norm distance, we study the RIP for subspaces in the top level in Section \ref{RIPofGaussian}. The main result of Theorem \ref{thm-rip} is proved by using two core propositions of Lemma \ref{lem-line-sub} and Theorem \ref{thm-sub}.
In Section \ref{preparation}, we focus on the probability and concentration inequalities of Gaussian random matrix to prepare necessary mathematical tools that will be used through this work.
In Section \ref{Appendixproof-thm-sub-lem-line-sub}, we prove the first core proposition of Lemma \ref{lem-line-sub}, which states that the affinity between a line and a subspace will concentrate on its estimate with overwhelming probability after random projection.
In Section \ref{Appendixproof-thm-sub}, we prove the second core proposition of Theorem \ref{thm-sub}, which provides a general theory that the affinity between two subspaces with arbitrary dimensions demonstrates concentration after random projection.
In Section \ref{Relatedworks}, we compare those theories with our previous results and highlight the novelty.
We conclude this work in Section \ref{secConclusion}.
Most proofs of lemmas and remarks are included in the Appendix \ref{secAppendix}.

\subsection{Notations}

Vectors and matrices are denoted by lower-case and upper-case letter, respectively, both in boldface.
${\bf A}^{\rm T}$ denotes matrix transposition.
$\|{\bf a}\|$ and $\|{\bf A}\|_{\rm F}$ denote $\ell_2$ norm of vector $\bf a$ and Frobenius norm of matrix $\bf A$.
$s_{\max}({\bf A})$ and $s_{\min}({\bf A})$ denote the largest and smallest singular value of matrix ${\bf A}$, respectively.
Subspaces are denoted by $\set X, \set Y,$ and $\set S$.
$\Span(\bf A)$ denotes the column space of matrix $\bf A$.
We use ${\set S}^\perp$ to denote the orthonormal complement space of $\set S$.
$\Proj_{\set S}({\bf v})$ denotes the projection of vector $\bf v$ onto subspace $\set S$.

\section{RIP of Gaussian Random Projection for Subspaces}
\label{RIPofGaussian}

\subsection{Preliminary}

Before starting the theoretical analysis, we first introduce the definition of principal angles and affinity.
These two concepts have been widely adopted to describe the relative position and to measure the similarity between two subspaces.
Our theoretical analysis will first focus on the estimation of these quantities before and after random projection.
Then using the connection between affinity and projection F-norm distance derived in \cite{li2017restricted}, we can readily derive the result in Theorem \ref{thm-rip}.

The principal angles (or canonical angles) between two subspaces provide a robust way to characterize the relative subspace positions \cite{jordan1875essai, PrincipalAngles2006}.

\begin{defi}\label{defi-principal-angles} The principal angles $\theta_1,\cdots,\theta_{d_1}$ between two subspaces ${\set X}_1$ and ${\set X}_2$ of dimensions $d_1\le d_2$, are recursively defined as
$$
\cos{\theta_k}=\max\limits_{{\bf x}_1 \in {\set X}_1}\max\limits_{{\bf x}_2 \in {\set X}_2}\frac{{\bf x}_1^{\rm T}{\bf x}_2}{\Vert {\bf x}_1\Vert \Vert {\bf x}_2\Vert}=:\frac{{\bf x}_{1k}^{\rm T}{\bf x}_{2k}}{\Vert {\bf x}_{1k}\Vert\Vert {\bf x}_{2k}\Vert},
$$
with the orthogonality constraints ${\bf x}_{i}^{\rm T}{\bf x}_{il}=0, l=1,\cdots,k-1,i=1,2$.
\end{defi}

Beside definition, an alternative way of computing principal angles is to use the singular value decomposition \cite{PrincipalAngles1973}.

\begin{lem}\label{lema-principal-angles-2}
Let the columns of ${\bf U}_i$ be orthonormal bases for subspace ${\set X}_i$ of dimension $d_i, i=1,2$ and suppose $d_1\le d_2$. Let $\lambda_1\ge\lambda_2\ge\cdots\ge\lambda_{d_1}\ge 0$ be the singular values of ${\bf U}_1^{\rm T}{\bf U}_2$, then
$
\cos\theta_k = \lambda_k,  k=1,\cdots, d_1.
$
\end{lem}

Based on principle angles, affinity is defined to measure the similarity between subspaces \cite{soltanolkotabi2012geometric}.

\begin{defi}\label{defi-affinity} The affinity between two subspaces ${\set X}_1$ and ${\set X}_2$ of dimension $d_1\le d_2$ is defined as
$$
\aff\left({\set X}_1, {\set X}_2\right) :=  \bigg(\sum_{k=1}^{d_1}\cos^2\theta_k\bigg)^{1/2}
=  \|{\bf U}_1^{\rm T}{\bf U}_2\|_{\rm F},
$$
where the columns of ${\bf U}_i$ are orthonormal bases of $\set{X}_i, i=1,2$.
\end{defi}

The relationship between distance and affinity is revealed in Lemma \ref{lem-aff-to-dist}.
Because of the concise definition and easy computation of affinity, we will start the theoretical analysis with affinity, and then present the results with distance by using Lemma \ref{lem-aff-to-dist}.

\begin{lem}\cite{li2017restricted}\label{lem-aff-to-dist}
The \emph{distance} and \emph{affinity} between two subspaces ${\set X}_1$ and ${\set X}_2$ of dimension $d_1, d_2$, are connected by
$$
 D^2(\set{X}_1, \set{X}_2) = \frac{d_1+d_2}{2}-\aff ^2(\set{X}_1, \set{X}_2).
$$
\end{lem}

\subsection{Theoretical Results}
In this section, we will present the main theoretical results about the affinity and distance between subspaces.
Before that, let us introduce some basic notations to be used.
We denote the random projection of subspaces of ${\set X}_1$ and ${\set X}_2$ as ${\set Y}_1$ and ${\set Y}_2$, respectively.
We denote $\dX = D({\set X}_1,{\set X}_2)$ and $\dY = D({\set Y}_1,{\set Y}_2)$ as the distances before and after random projection.
Similarly, we use $\affX = \aff({\set X}_1,{\set X}_2)$ and $\affY = \aff({\set Y}_1,{\set Y}_2)$ to denote the affinities before and after projection.
Without loss of generality, we always suppose that $d_1\le d_2$.
For simplicity, we refer the affinity (distance) after random projection as \emph{projected affinity} (\emph{projected distance}).

To begin with, we focus on a special case that one subspace is degenerated to a line (one-dimensional subspace).
The following lemma provides an estimation of the affinity between a line and a subspace after Gaussian random projection.
When the dimensionality of the new ambient space is large enough, the real projected affinity will highly concentrate around this estimation with overwhelming probability.

\begin{lem}\label{lem-line-sub}
Suppose $\set X_1, \set X_2\subset \mathbb{R}^N$ are a line and a $d$-dimension subspace, $d\ge 1$, respectively.
Let $\lambda = \aff_{\set X}$ denote their affinity.
If they are projected onto $\mathbb R^n, n<N,$ by a Gaussian random matrix ${\bm\Phi} \in\mathbb{R}^{n\times N}$, $\set{X}_i \stackrel{\bm \Phi}{\longrightarrow} \set{Y}_i, i=1,2$, then the projected affinity, $\affY$, can be estimated by
\begin{equation}\label{lem-line-sub-aff-change}
\oaffYS = \lambda^2+\frac{d}{n}\left(1-\lambda^2\right),
\end{equation}
and there exist constants $c_1(\varepsilon),c_2(\varepsilon) > 0$ depending only on $\varepsilon$ such that for any $n > c_1(\varepsilon)d$,
\begin{equation}\label{lem-line-sub-aff-bound}
    \left|\affYS-\oaffYS\right| < (1-\lambda^2)\varepsilon
\end{equation}
holds with probability at least $1 - {\rm e}^{-c_2(\varepsilon)n}$.
\end{lem}

Then, we study the general case of projecting two subspaces of arbitrary dimensions.
As mentioned in the last subsection, we will begin with the estimation of affinity and then restate the result in terms of distance.

The following theorem reveals the concentration of affinity between two arbitrary subspaces after random projection.

\begin{thm}\label{thm-sub}
Suppose $\set X_1, \set X_2\subset \mathbb{R}^N$ are two subspaces with dimension $d_1 \le d_2$, respectively.
Take
\begin{equation}\label{thm-sub-aff-change}
\oaffYS = \affXS+\frac{d_2}{n}(d_1-\affXS)
\end{equation}
as an estimate of the affinity between two subspaces after random projection, $\set{X}_i \stackrel{\bm \Phi}{\longrightarrow} \set{Y}_i, i=1,2$.
Then there exist constants $c_1(\varepsilon),c_2(\varepsilon) > 0$ depending only on $\varepsilon$ such that for any $n > c_1(\varepsilon)d_2$,
\begin{equation}\label{thm-sub-aff-bound}
    \left|\affYS-\oaffYS\right| < (d_1-\affXS)\varepsilon
\end{equation}
holds with probability at least $1 - {\rm e}^{-c_2(\varepsilon)n}$.
\end{thm}

Because of its concision when evaluating the relative position in Definition \ref{defi-affinity}, we present the concentration by using affinity in Lemma \ref{lem-line-sub} and Theorem \ref{thm-sub}, which play essential role in RIP for subspaces.
Their proofs, which unfurl main text of this work, are postponed to  Section \ref{Appendixproof-thm-sub-lem-line-sub} and Section \ref{Appendixproof-thm-sub}, respectively.

Using Lemma \ref{lem-aff-to-dist} and Theorem \ref{thm-sub}, we derive an estimation of the projected distance.
Similarly we prove that the true projected distance will highly concentrate around this estimate with overwhelming probability.

\begin{cor}\label{thm-sub-dis}
Suppose $\set X_1, \set X_2\subset \mathbb{R}^N$ are two subspaces with dimension $d_1 \le d_2$, respectively.
We use
\begin{equation}\label{thm-sub-dis-change}
\odYS = \dXS-\frac{d_2}{n}\left(\dXS - \frac{d_2-d_1}{2}\right)
\end{equation}
as an estimation of the distance between two subspaces after random projection, $\set{X}_i \stackrel{\bm \Phi}{\longrightarrow} \set{Y}_i, i=1,2$.
Then there exist constants $c_1(\varepsilon),c_2(\varepsilon) > 0$ depending only on $\varepsilon$ such that for any $n > c_1(\varepsilon)d_2$,
\begin{equation}\label{thm-sub-dis-bound}
    \left|\dYS-\odYS\right| < \left(\dXS - \frac{d_2-d_1}{2}\right)\varepsilon
\end{equation}
holds with probability at least $1 - {\rm e}^{-c_2(\varepsilon)n}$.
\end{cor}
\begin{proof}
Combining \eqref{thm-sub-aff-change} and \eqref{thm-sub-dis-change} by using Lemma \ref{lem-aff-to-dist}, we readily get that $|\dYS-\odYS| = |\affYS-\oaffYS|$.
Using Lemma \ref{lem-aff-to-dist} again, we have
$
\dXS - \frac{d_2-d_1}{2}=d_1-\affXS.
$
Therefore, \eqref{thm-sub-dis-bound} is identical to \eqref{thm-sub-aff-bound}.
\end{proof}

\subsection{Proof of Theorem \ref{thm-rip}}
\label{Appendixproof-thm-sub-multi}

Now we are ready to prove the RIP of Gaussian random matrix for projecting a finite set of subspaces using the results above.

Without loss of generality, we assume that $d_i\le d_j\le d$.
According to Corollary \ref{thm-sub-dis}, there exist constants $c_{1,1}$, $c_{2,1}>0$ depending only on $\varepsilon$ such that for any $n>c_{1,1}d_j$,
\begin{align*}
D^2({\set X}_i,{\set X}_i)-\left(\frac{d_j}{n}+\frac{\varepsilon}{2}\right)&\left(D^2({\set X}_i,{\set X}_i)-\frac{d_j-d_i}{2}\right)<D^2({\set Y}_i,{\set Y}_i)\\
&<D^2({\set X}_i,{\set X}_i)+\left(-\frac{d_j}{n}+\frac{\varepsilon}{2}\right)\left(D^2({\set X}_i,{\set X}_i)-\frac{d_j-d_i}{2}\right).
\end{align*}
holds with probability at least $1-{\rm e}^{-c_{2,1}n}$.
When $n>2d/\varepsilon$, we have ${d_j}/{n}\le{d}/{n}<\varepsilon/2$.
In this case, we have both
\begin{align}
D^2({\set Y}_i,{\set Y}_i)&>D^2({\set X}_i,{\set X}_i)-\left(\frac{d_j}{n}+\frac{\varepsilon}{2}\right)D^2({\set X}_i,{\set X}_i)\nonumber\\
&=\left(1-\frac{d_j}{n}-\frac{\varepsilon}{2}\right)D^2({\set X}_i,{\set X}_i)>\left(1-\varepsilon\right)D^2({\set X}_i,{\set X}_i),\label{eq-mainproof-1}\\
D^2({\set Y}_i,{\set Y}_i)&<D^2({\set X}_i,{\set X}_i)+\left(-\frac{d_j}{n}+\frac{\varepsilon}{2}\right)D^2({\set X}_i,{\set X}_i)\nonumber\\
&=\left(1-\frac{d_j}{n}+\frac{\varepsilon}{2}\right)D^2({\set X}_i,{\set X}_i)<\left(1+\varepsilon\right)D^2({\set X}_i,{\set X}_i),\label{eq-mainproof-2}
\end{align}
hold with probability at least $1-{\rm e}^{-c_{2,1}n}$.
Note that \eqref{eq-mainproof-1} and \eqref{eq-mainproof-2} hold for any $1\le i<j\le L$.
Then the probability is at least $1-\frac{L(L-1)}{2}{\rm e}^{-c_{2,1}n}$.
If $n>\frac{1}{c_{2,1}}\ln \frac{L(L-1)}{2}$, there exists constant $c_2$ depending only on $\varepsilon$, such that
$
\frac{L(L-1)}{2}{\rm e}^{-c_{2,1}n} < {\rm e}^{-c_{2}n}.
$
Take $c_1:=\max\{c_{1,1},\frac{2}{\varepsilon},\frac{2}{c_{2,1}}\}$, then when $n>c_1\max\{d,\ln L\}$, conditions $n>c_{1,1}d$, $n>2d/\varepsilon$, and $n > \frac{1}{c_{2,1}}\ln \frac{L(L-1)}{2}$ that are required above are all satisfied, the probability is at least $1-{\rm e}^{-c_2n}$ with $c_2>0$.
Then we reach the final conclusion.

\section{Concentration Inequalities for Gaussian Distribution}
\label{preparation}
Before proving the main results, we first introduce some useful concentration inequalities for Gaussian distribution.
Most of them are proved using the following lemma, which provides a strict estimation of the singular values of Gaussian random matrix.

\begin{lem}\cite{Davidson2001Local}\label{T}
Let ${\bf A}$ be an $N \times n$ matrix whose elements $a_{ij}$ are independent Gaussian random variables. Then for every $t \ge 0$, one has
\begin{align}\label{T-probability1}
\mathbb{P}\left(s_{\max}({\bf A}) \ge \sqrt{N} + \sqrt{n} + t\right) \le {\rm e}^{-\frac{t^2}2},
\end{align}
and
\begin{align}\label{T-probability2}
\mathbb{P}\left(s_{\min}({\bf A}) \le \sqrt{N} - \sqrt{n} - t\right) \le {\rm e}^{-\frac{t^2}2}.
\end{align}
\end{lem}

Based on Lemma \ref{T}, we are ready to prove some useful lemmas that will be directly used to prove our theories on RIP of random projection for subspaces. Before doing that, we first define standard Gaussian random matrix and verify that the function satisfying certain condition can be written as a single exponential function.

\begin{defi}\label{defi-guassian}
A Gaussian random matrix (or vector) has \emph{i.i.d.} zero-mean Gaussian random entries.
A standard Gaussian random matrix ${\bf A} \in\mathbb{R}^{n\times N}$ has \emph{i.i.d.} zero-mean Gaussian random entries with variance ${1}/{n}$.
Each column of $\bf A$ is a standard Gaussian random vector.

\end{defi}

\begin{lem}\label{exp-formal}
Given
\begin{align}
f(\varepsilon, n, \tau) = \frac{1}{K}\sum_{k = 1}^K a_k(\varepsilon, n){\rm e}^{-g_k(\varepsilon, n, \tau)},
\end{align}
if for all $k$, it holds that
\begin{align}
h_k(\varepsilon) &:= \lim_{\tau \to 0}\lim_{n \to \infty} \frac{g_k(\varepsilon, n, \tau)}n > 0,\label{exp-formal-condition1}\\
b_k(\varepsilon) &:=\lim_{n \to \infty} \frac{\ln a_k(\varepsilon, n)}{n} < h_k(\varepsilon),\label{exp-formal-condition2}
\end{align}
then there exist universal constants $n_0, c_1 > 0,$ and $c_2 > 0$ depending only on $\varepsilon$, such that when $n > n_0, \tau < c_1$, it satisfies that
$
f(\varepsilon, n, \tau) < {\rm e}^{-c_2 n}.
$
\end{lem}
\begin{proof}
The proof is postponed to Appendix \ref{proof-exp-formal}.
\end{proof}

\begin{rem}\label{exp-summation}
Lemma \ref{exp-formal} illustrates that the summation of finite multiple exponential decay functions can always be bounded by a single exponential function.
\end{rem}

The following lemma illustrates that the norm of standard Gaussian random vector concentrates around $1$ with high probability, especially when the dimensionality is high.

\begin{lem}\label{lem2}
Assume that $\va\in\mathbb{R}^n$ is a standard Gaussian random vector.
For any $\varepsilon>0$, we have
\begin{align}\label{lem-eq1}
\mathbb{P}\left(\left|\|\va\|^2-1\right|>\varepsilon\right)<{\rm e}^{-c(\varepsilon)n}
\end{align}
hold for $n>n_0$, where $n_0$ and $c$ are constants dependent on $\varepsilon$.
\end{lem}
\begin{proof}
The proof is postponed to Appendix \ref{proof-lem2}.
\end{proof}

Furthermore, Corollary \ref{cor-lem2}  generalizes Lemma \ref{lem2} to case when we project standard Gaussian random vector by orthonormal matrix.

\begin{cor}\label{cor-lem2}
Let $\va \in\mathbb{R}^n$ be a standard Gaussian random vector.
For any given orthonormal matrix ${\bf V}=\left[{\bf v}_1,\cdots,{\bf v}_d\right]\in\mathbb{R}^{n\times d}$ and $\varepsilon>0$, we have
\begin{equation}\label{eq-p2epsilon1}
    \mathbb{P}\left(\left|\|{\bf V}^{\rm T}{\bf a}\|^2 - \frac{d}{n}\right|>\varepsilon\right)<{\rm e}^{-c_2(\varepsilon)n}
\end{equation}
hold for $n>c_1d$, where $c_1$, $c_2$ are constants dependent on $\varepsilon$.
\end{cor}

\begin{proof}
The proof is postponed to Appendix \ref{proof-cor-lem2}.
\end{proof}

Corollary \ref{unit-singular-value} extends Lemma \ref{T} to column-normalized standard Gaussian random matrix.
It reveals that a column-normalized Gaussian random matrix is a high-quality approximation to an orthonormal matrix,
because all its singular values are very close to $1$.

\begin{cor}\label{unit-singular-value}
Let ${\bf A}=\left[{\bf a}_1,\cdots,{\bf a}_k\right] \in \mathbb{R}^{n\times k}$ be a standard Gaussian random matrix.
Each column of $\bar{\bf A}$ is normalized from the corresponding column of ${\bf A}$, that is
$$
\bar{\bf A}=\left[\frac{{\bf a}_1}{\left\|{\bf a}_1\right\|},\cdots,\frac{{\bf a}_k}{\left\|{\bf a}_k\right\|}\right].
$$
Then we can get the bound of the minimum and maximum of the singular value of $\bar{\bf A}$ as below
\begin{align}
&\mathbb{P}\left(s_{\min}^2\left(\bar{\bf A}\right)<1-\varepsilon\right)<{\rm e}^{-c_{2,1}(\varepsilon)n},\quad \forall n>c_{1,1}k, \label{eq27}\\
&\mathbb{P}\left(s_{\max}^2\left(\bar{\bf A}\right)>1+\varepsilon\right)< {\rm e}^{-c_{2,2}(\varepsilon)n}, \quad \forall n>c_{1,2}k,\label{eq28}
\end{align}
where
$c_{1,1}$ and $c_{2,1}$, $c_{1,2}$ and $c_{2,2}$ are constants dependent on $\varepsilon$ in \eqref{eq27} and \eqref{eq28}, respectively.
\end{cor}
\begin{proof}
The proof is postponed to Appendix \ref{proof-unit-singular-value}.
\end{proof}

\begin{rem}\label{P4-nd}
For ${\bf A}\in \mathbb{R}^{(n-d_0)\times k}$ be a standard Gaussian random matrix, we have \eqref{eq28} hold for $n>c_{1,2}\max\{k,d_0\}$, where $c_{1,2}$ and $c_{2,2}$ are constants dependent on $\varepsilon$.
\end{rem}

The following lemma studies a property of Gaussian random projection.
Intuitively, it illustrates that if a line and a subspace are perpendicular to each other, they will still be almost perpendicular after Gaussian random projection.

\begin{lem}\label{lem5}
Assume ${\bf u}_1\in\mathbb{R}^N$ is a unit vector, ${\bf U}_2\in\mathbb{R}^{N\times d}$ is an orthonormal matrix, and ${\bf u}_1$ is perpendicular to ${\bf U}_2$.
Let ${\bm \Phi}\in\mathbb{R}^{n\times N}$ be a standard Gaussian random matrix.
We use ${\bf a}_1 = {\bm \Phi}{\bf u}_1$ and ${\bf A}_2 = {\bm \Phi}{\bf U}_2$ to denote the projection of ${\bf u}_1$ and ${\bf U}_2$ by using $\bm\Phi$.
If ${\bf V}_2$ is an arbitrary orthonormal basis of $\Span({\bf A}_2)$, then for $\varepsilon>0$, we have
\begin{equation}\label{eq-p2epsilon1}
\mathbb{P}\left(\left\|{\bf V}_2^{\rm T}{\bf a}_1\right\|^2>\varepsilon\right) \le {\rm e}^{-c_2(\varepsilon)n}
\end{equation}
hold for $n>c_1(\varepsilon)d$, where $c_1$, $c_2$ are constants determined by $\varepsilon$.
\end{lem}
\begin{proof}
The proof is postponed to Appendix \ref{proof-lem5}.
\end{proof}

\begin{cor}\label{cor-lem5}
In Lemma \ref{lem5}, if we further define $\bar{\bf a}_1 = {\bf a}_1/\|{\bf a}_1\|$ as the normalized projection of ${\bf u}_1$,
then for $\varepsilon>0$, we have
\begin{equation}\label{eq-p2epsilon1new}
\mathbb{P}\left(\left\|{\bf V}_2^{\rm T}\bar{\bf a}_1\right\|^2>\varepsilon\right) \le {\rm e}^{-c_2(\varepsilon)n}
\end{equation}
hold for $n>c_1(\varepsilon)d$, where $c_1$, $c_2$ are constants determined by $\varepsilon$.
\end{cor}
\begin{proof}
The proof is postponed to Appendix \ref{proof-cor-lem5}.
\end{proof}

\begin{rem}\label{rem-p2epsilon1new}
Using the same notations in Corollary \ref{cor-lem5}, if we take ${\bm \Phi}\in\mathbb{R}^{(n-d_0)\times N}$, $\bar{\bf a}_1\in\mathbb{R}^{n-d_0}$ and ${\bf V}_2\in\mathbb{R}^{(n-d_0)\times d}$, then we have still have
$$
\mathbb{P}\left(\left\|{\bf V}_2^{\rm T}\bar{\bf a}_1\right\|^2>\varepsilon\right) \le {\rm e}^{-c_2(\varepsilon)n}
$$
for $n>c_{1}(\varepsilon)\max\{d,d_0\}$, where $c_1$ and $c_2$ are constants dependent on $\varepsilon$.
This can by readily verified by replacing $n$ with $n-d_0$ in \eqref{eq-p2epsilon1new} and then applying Lemma \ref{exp-formal}.
The detailed proof is postponed to Appendix \ref{proof-rem-p2epsilon1new}.
\end{rem}

Finally we state the independence of random vectors, matrices, and their column-spanned subspaces.

\begin{defi}\label{defi-independence}
Two random vectors are independent, if and only if the distribution of any one of them does not have influence on that of the other.
Two random matrices are independent, if and only if any two columns of them are independent.
Furthermore, we introduce the independence between a random matrix and a subspace, which holds true if and only if the subspace is spanned by the columns of another random matrix that is independent of the first one.
Finally, two subspaces are independent, if and only if they are spanned by the columns of two independent random matrices, respectively.
\end{defi}

\begin{lem}\label{independence}
Assume ${\bf U}$ and ${\bf V}$ are two matrices satisfying ${\bf U}^{\rm T}{\bf V}={\bf 0}$, and $\bm\Phi$ is a Gaussian random matrix.
Then ${\bm\Phi}{\bf U}$ and ${\bm\Phi}{\bf V}$ are independent.
This can be readily verified by calculating the correlation between any two entries in ${\bm\Phi}{\bf U}$ and ${\bm\Phi}{\bf V}$, respectively.
They are all zero.
\end{lem}

\section{Proof of Lemma \ref{lem-line-sub}}
\label{Appendixproof-thm-sub-lem-line-sub}

The proof of Lemma \ref{lem-line-sub} is made up of two steps.
At first, we derive an accurate expression of the error about estimating $\affY$.
Then, the estimate error is bounded by utilizing the concentration inequalities for Gaussian distribution that we derived in the previous section.

\emph{\bf Step 1)}
Let us begin from choosing the bases for the line ${\set X}_1$ and the subspace ${\set X}_2$ and then calculating the affinity after projection.

According to the definition of affinity, $\lambda=\cos\theta$, where $\theta$ is the only principal angle between ${\set X}_1$ and ${\set X}_2$.
We use $\bf u$ and ${\bf u}_1$ to denote, respectively, the basis of ${\set X}_1$ and a unit vector in ${\set X}_2$, which constructs the principal angle with $\bf u$.
Therefore, $\bf u$ can be rewritten into the following form
\begin{align}\label{decompose-u}
{\bf u}=\lambda{\bf u}_1+\sqrt{1-\lambda^2}{\bf u}_0,
\end{align}
where ${\bf u}_0$ denotes some unit vector orthogonal to ${\set X}_2$.
Based on the above definition, we can choose ${\bf U}=\left[{\bf u}_1, ..., {\bf u}_d\right]$ as the basis of ${\set X}_2$.
Notice that $\{{\bf u}_2,\cdots,{\bf u}_d\}$ could be freely chosen as long as the orthonormality is satisfied.

After projecting ${\set X}_1$ by random Gaussian matrix, we get subspace ${\set Y}_1$, whose basis vector is
\begin{align}
{\bf a}={\bm \Phi}{\bf u} &= \lambda{\bm \Phi}{\bf u}_1+\sqrt{1-\lambda^2}{\bm \Phi}{\bf u}_0\nonumber\\
&=\lambda{\bf a}_1+\sqrt{1-\lambda^2}{\bf a}_0,\label{eq-lem-line-sub-proof-a}
\end{align}
where ${\bf a}_1:={\bm\Phi}{\bf u}_1$ and ${\bf a}_0:={\bm\Phi}{\bf u}_0$ are not orthogonal to each other.
As for ${\set Y}_2$, considering that ${\bm\Phi}{\bf U}$ is not a set of orthonormal basis, we do orthogonalization by using Gram-Schmidt process.
Denote the orthonormalized matrix as ${\bf V} = \left[{\bf v}_1,\cdots,{\bf v}_d\right]$. By the definition of Gram-Schmidt process, the first column of ${\bf V}$ should be
\begin{equation}\label{eq-lem-line-sub-proof-v1}
{\bf v}_1=\frac{{\bf a}_1}{\|{\bf a}_1\|},
\end{equation}
which does not change its direction after the orthogonalization.

\begin{rem}\label{principalangle-affinity-projection}
Consider the affinity between two subspaces $\mathcal{S}_1$ and $\mathcal{S}_2$, with dimension 1 and $d\ge1$.
Let ${\bf v}_1$ and ${\bf V}_2$ be the orthonormal basis for $\mathcal{S}_1$ and $\mathcal{S}_2$, respectively.
Then the affinity equals the norm of the projection of ${\bf v}_1$ onto $\mathcal{S}_2$, i.e., $\lambda = \|\Proj_{{\set S}_2}({\bf v}_1)\|=\left\|{\bf V}_2^{\rm T}{\bf v}_1\right\|$.
\end{rem}

By the definition of affinity and \eqref{eq-lem-line-sub-proof-a}, we can calculate the affinity between ${\set Y}_1$ and ${\set Y}_2$ as
\begin{align}\label{eq-lem-line-sub-proof-affY2}
        \affYS &=\left\|{\bf V}^{\rm T}\frac{\bf a}{\|{\bf a}\|}\right\|^2\nonumber\\
        &= \frac1{\|{\bf a}\|^2}\left\|\lambda {\bf V}^{\rm T}{\bf a}_1 + \sqrt{1-\lambda^2}{\bf V}^{\rm T}{\bf a}_0\right\|^2\nonumber\\
        &= \frac1{\|{\bf a}\|^2}\left(\lambda^2\|{\bf V}^{\rm T}{\bf a}_1\|^2 + (1-\lambda^2)\|{\bf V}^{\rm T}{\bf a}_0\|^2 + \lambda\sqrt{1-\lambda^2}\|{\bf a}_1^{\rm T}{\bf a}_0\| \right).
\end{align}
Because ${\bf a}_1$ lies in ${\set Y}_2$, we have
\begin{equation}\label{eq-lem-line-sub-proof-a1}
\|{\bf V}^{\rm T}{\bf a}_1\| = \|{\bf a}_1\|.
\end{equation}
By taking the norm on both sides of \eqref{eq-lem-line-sub-proof-a}, we write
\begin{align}
\|{\bf a}\|^2 = \!\lambda^2\|{\bf a}_1\|^2+ (1-\lambda^2)\|{\bf a}_0\|^2 + 2\lambda\sqrt{1-\lambda^2}\|{\bf a}_0\|\|{\bf a}_1\|. \label{eq-lem-line-sub-proof-a2}
\end{align}
Eliminating $\|{\bf V}^{\rm T}{\bf a}_1\|$ and $\|{\bf a}_1\|$ by inserting \eqref{eq-lem-line-sub-proof-a1} and \eqref{eq-lem-line-sub-proof-a2} into \eqref{eq-lem-line-sub-proof-affY2}, we get
\begin{align}
        \affYS &=\!\frac{1}{\|{\bf a}\|^2}\!\left(\!\|{\bf a}\|^2 \!- \!(1\!-\!\lambda^2)\|{\bf a}_0\|^2
        \!+\! (1\!-\!\lambda^2)\|{\bf V}^{\rm T}{\bf a}_0\|^2\right) \nonumber\\
        &=1 - (1-\lambda^2)\left(\frac{\|{\bf a}_0\|^2}{\|{\bf a}\|^2} - \frac{\|{\bf V}^{\rm T}{{\bf a}_0}\|^2}{\|{\bf a}\|^2}\right).\label{eq-lem-line-sub-proof-affY-2}
\end{align}
Recalling \eqref{lem-line-sub-aff-change} and inserting the estimation into \eqref{eq-lem-line-sub-proof-affY-2}, the estimate error is deduced as
\begin{align}
\left|\affYS-\oaffYS\right|
= &\left|\affYS-\left(\lambda^2+\frac{d}{n}(1-\lambda^2)\right)\right| \nonumber\\
=& \left|1 \!-\! (1\!-\!\lambda^2)\left(\frac{\|{\bf a}_0\|^2}{\|{\bf a}\|^2} - \frac{\|{\bf V}^{\rm T}{{\bf a}_0}\|^2}{\|{\bf a}\|^2}\right)\!-\!\left(\!1\!-\!(1\!-\!\lambda^2)\!\left(1\!-\!\frac{d}{n}\right)\!\!\right)\right| \nonumber\\
=& (1-\lambda^2)\left|\left(1-\frac{d}{n}\right)-\left(\frac{\|{\bf a}_0\|^2}{\|{\bf a}\|^2} - \frac{\|{\bf V}^{\rm T}{{\bf a}_0}\|^2}{\|{\bf a}\|^2}\right)\right|. \label{eq-lem-line-sub-proof-esterr}
\end{align}

\emph{\bf Step 2)}
Before bounding the estimate error by concentration inequalities,
we first split the RHS of \eqref{eq-lem-line-sub-proof-esterr} into three parts using triangle inequality.
\begin{align}
& \left|\left(1-\frac{d}{n}\right)-\left(\frac{\|{\bf a}_0\|^2}{\|{\bf a}\|^2} - \frac{\|{\bf V}^{\rm T}{{\bf a}_0}\|^2}{\|{\bf a}\|^2}\right)\right|\nonumber\\
\le& \left|\frac{\|{\bf a}_0\|^2}{\|{\bf a}\|^2}-1\right| + \left|\frac{\|{\bf V}^{\rm T}{{\bf a}_0}\|^2}{\|{\bf a}\|^2} - \frac{d}{n}\right| \nonumber\\
\le& \left|\frac{\|{\bf a}_0\|^2}{\|{\bf a}\|^2}-1\right| +
\frac{1}{\|\va\|^2}\left|\|{\bf V}^{\rm T}{\bf a}_0\|^2 - \frac{d}{n}\right|+
\frac{d}{n}\left|\frac{1}{\|\va\|^2}-1\right|.\label{eq-lem-line-sub-proof-decomp-error}
\end{align}

Since both $\bf a$ and ${\bf a}_0$ are standard Gaussian random vectors, by Lemma \ref{lem2}, for $n>n_0$, with probability at least $1 - 2{\rm e}^{-c_{2,1}(\varepsilon_1)n}$, we have
\begin{equation}\label{eq-lem-line-sub-proof-concen1}
\left|\|{\bf a}\|^2 - 1\right| < \varepsilon_1
\end{equation}
and
\begin{equation}\label{eq-lem-line-sub-proof-concen2}
\left|\|{\bf a}_0\|^2 - 1\right| < \varepsilon_1.
\end{equation}
Since ${\bf u}_0$ is orthogonal to ${\bf u}_i$, $i=1,\cdots,d$, by Corollary \ref{cor-lem2}, for $n>c_{1,2}(\varepsilon_2)d$,
with probability at least $1 - {\rm e}^{-c_{2,2}(\varepsilon_2)n}$, we have
\begin{equation}\label{eq-lem-line-sub-proof-concen3}
\left|\left\|{\bf V}^{\rm T}{\bf a}_0\right\|^2 -\frac{d}n\right|
<\varepsilon_2.
\end{equation}
Using \eqref{eq-lem-line-sub-proof-concen1}, \eqref{eq-lem-line-sub-proof-concen2}, and
\eqref{eq-lem-line-sub-proof-concen3} in \eqref{eq-lem-line-sub-proof-decomp-error}, for $n>\max\{n_0,c_{1,2}\}d$,
with probability at least
\begin{align}\label{eq-thm-sub-aff-probability}
1 - 2{\rm e}^{-c_{2,1}(\varepsilon_1)n} - {\rm e}^{-c_{2,2}(\varepsilon_2)n},
\end{align}
we have
\begin{align}
\left|\left(1-\frac{d}{n}\right)-\frac{\|{\bf a}_0\|^2}{\|{\bf a}\|^2}\left(1 - \sum_{i=1}^d \cos^2\theta_i\right)\right|
\le&
\frac{2\varepsilon_1}{1-\varepsilon_1}+\frac{\varepsilon_2}{1-\varepsilon_1} + \frac{d}{n}\frac{\varepsilon_1}{1-\varepsilon_1}\nonumber\\
\le& \left(\frac{3d}{2n}+3\right)\varepsilon_1 +\frac{3\varepsilon_2}{2},
\label{eq-lem-line-sub-proof-esterrRHS-pre}
\end{align}
where the last inequality holds for $\varepsilon_1 < 1/3$.

To complete the proof, we need to formulate \eqref{eq-lem-line-sub-proof-esterrRHS-pre} and \eqref{eq-thm-sub-aff-probability} into the shape of \eqref{lem-line-sub-aff-bound} and a single exponential function, respectively.
Letting $\varepsilon_1=\varepsilon_2=:\varepsilon/6$ and inserting \eqref{eq-lem-line-sub-proof-esterrRHS-pre} into \eqref{eq-lem-line-sub-proof-esterr}, we have
\begin{align*}
\left|\affYS-\oaffYS\right| &\le\left(1-\lambda^2\right)\left(\left(\frac{3d}{2n}+3\right)\varepsilon_1 +\frac{3\varepsilon_2}{2}\right)\\
&\le\left(1-\lambda^2\right)\left(\left(\frac{3}{2}+3\right)\varepsilon_1 +\frac{3\varepsilon_2}{2}\right)\\
&=\left(1-\lambda^2\right)\varepsilon.
\end{align*}
 By Remark \ref{exp-summation}, we know that there exist constants $c_1$, $c_2$, such that \eqref{eq-thm-sub-aff-probability} is greater than $1-{\rm e}^{-c_2(\varepsilon)n}$ for any $n>c_1(\varepsilon)d$ and close the proof.

\section{Proof of Theorem \ref{thm-sub}}
\label{Appendixproof-thm-sub}

The proof of Theorem \ref{thm-sub} is divided into two parts.
In subsection \ref{subsectionA}, we will complete the main body of the proof by using an important lemma,
which will be proved in subsection \ref{subsectionB}.
Before we start, let us introduce some auxiliary variables.

\begin{rem}\label{quasi-construct}
Assume there are two subspaces $\mathcal{S}_1$ and $\mathcal{S}_2$, with dimension $d_1\le d_2$.
Let $\tilde{\bf U}_i=\left[\tilde{\bf u}_{i,1}, \cdots, \tilde{\bf u}_{i, d_i}\right]$ denote any orthonormal matrix for subspace ${\set S}_i, i=1,2$.
One may do singular value decomposition as
$
\tilde{\bf U}_2^{\rm T}\tilde{\bf U}_1 = {\bf Q}_2{\bm\Lambda}{\bf Q}_1^{\rm T}
$,
where the singular values $\lambda_k=\cos\theta_k, 1\le k\le d_1$ are located on the diagonal of $\bm\Lambda$,
and $\theta_1\le \theta_2 \le \cdots\le \theta_{d_1}$ denote the principal angles between $\mathcal{S}_1$ and $\mathcal{S}_2$. After reshaping, we have
\begin{align*}
{\bf U}_2^{\rm T}{\bf U}_1 := \left(\tilde{\bf U}_2{\bf Q}_2\right)^{\rm T}\tilde{\bf U}_1{\bf Q}_1 = {\bm\Lambda}
= \left[ \begin{array}{ccc}
     \lambda_1 & & \\ & \ddots&  \\ & & \lambda_{d_1} \\ \hline  & {\bf 0} &  \end{array} \right],
\end{align*}
where
$$
{\bf U}_i := \tilde{\bf U}_i{\bf Q}_i = \left[{\bf u}_{i,1},\cdots,{\bf u}_{i,d_i}\right], \quad i=1,2
$$
are the orthonormal basis, which have the closest connection with the affinity between these two subspaces.
\end{rem}

\begin{defi}[\bf \emph{principal} orthonormal bases]\label{quasi}
We refer ${\bf U}_1$ and ${\bf U}_2$ as \emph{principal} orthonormal bases for $\mathcal{S}_1$ and $\mathcal{S}_2$,
if they are derived by using the method in Remark \ref{quasi-construct}.
\end{defi}

According to Remark \ref{quasi-construct}, for subspaces ${\set X}_1$ and ${\set X}_2$, we can get their \emph{principal} orthonormal bases ${\bf U}_1$ and ${\bf U}_2$, respectively.
After projection by multiplying a standard Gaussian random matrix ${\bm\Phi}$, the original basis matrix changes to
$
{\bf A}_i={\bm\Phi}{\bf U}_i=\left[{\bf a}_{i,1}, \cdots, {\bf a}_{i,d_i}\right],
$
whose columns are no longer unitary and orthogonal to each other.  Then we normalize each columns as
$
\bar{\bf A}_i=\left[\bar{\bf a}_{i,1}, \cdots, \bar{\bf a}_{i,d_i}\right]=\left[\frac{{\bf a}_{i,1}}{\|{\bf a}_{i,1}\|}, \cdots, \frac{{\bf a}_{i,d_i}}{\|{\bf a}_{i,d_i}\|}\right],
$
whose columns are now unitary but still not orthogonal to each other.
However, by Corollary \ref{unit-singular-value}, we know that $\bar{\bf A}_i$ can be used as a good approximation for the orthonormal basis of ${\set Y}_i$.
We will see that $\bar{\bf A}_i$ plays an important role in estimating the affinity after projection.

We also need to define an accurate orthonormal basis for the projected subspace.
One efficient way is to process $\bar{\bf A}_i$ by using Gram-Schmidt orthogonalization,
whose result is defined as
$
{\bf V}_i=\left[{\bf v}_{i,1}, \cdots, {\bf v}_{i,d_i}\right], i=1,2.
$

\subsection{Main Body}\label{proof1}
\label{subsectionA}

In order to prove Theorem \ref{thm-sub}, we need to calculate $\affYS$ and estimate its bias from $\oaffYS$.
Because $\bar{\bf A}_i$ is very close to an orthonormal matrix, we may use $\|{\bf V}_2^{\rm T}\bar{\bf A}_1\|$ to estimate the affinity after projection.
By using triangle inequality, we have
\begin{equation}\label{eq-proof-thm-sub-decomp}
\left|\affYS-\oaffYS\right|
\le \left|\affYS-\left\|{\bf V}_2^{\rm T}\bar{\bf A}_1\right\|_{\rm F}^2\right| + \left|\left\|{\bf V}_2^{\rm T}\bar{\bf A}_1\right\|_{\rm F}^2-\oaffYS\right|.
\end{equation}
Therefore, the following proof can be divided into three steps.
The first step is to bound the error caused by using $\bar{\bf A}_1$ as an approximation of ${\bf V}_1$ to compute the affinity.
To do that, we will introduce an important lemma, which is the essence of the proof.
The second step is to bound the difference between the approximated affinity and our estimate, which can be derived by using Lemma \ref{lem-line-sub}.
Finally, we combine these two bounds and complete the proof.

{\bf Step 1)}
For the first item in the RHS of \eqref{eq-proof-thm-sub-decomp}, according to the definition of affinity, we have
\begin{align}
\left|\affYS-\left\|{\bf V}_2^{\rm T}\bar{\bf A}_1\right\|_{\rm F}^2\right| &= \left|\|{\bf V}_2^{\rm T}{\bf V}_1\|_{\rm F}^2-\left\|{\bf V}_2^{\rm T}\bar{\bf A}_1\right\|_{\rm F}^2\right|  \nonumber\\
&= \left|\sum_{k=1}^{d_1} \left(\|{\bf V}_2^{\rm T}{\bf v}_{1, k}\|^2-\left\|{\bf V}_2^{\rm T}\bar{\bf a}_{1, k}\right\|^2\right) \right|  \nonumber\\
&\le \sum_{k=1}^{d_1} \left|\|{\bf V}_2^{\rm T}{\bf v}_{1, k}\|^2-\left\|{\bf V}_2^{\rm T}\bar{\bf a}_{1, k}\right\|^2\right|.\label{eq-diff1-sum}
\end{align}

\begin{lem}\label{vTV2-aTV2}
There exist constants $c_{1,1}(\varepsilon_1)$ and $c_{2,1}(\varepsilon_1)>0$ depending only on $\varepsilon_1$, such that for any $n>c_{1,1}(\varepsilon_1)d_2$, we have
\begin{align}\label{eq-vTV2-aTV2}
\left|\|{\bf V}_2^{\rm T}{\bf v}_{1, k}\|^2-\left\|{\bf V}_2^{\rm T}\bar{\bf a}_{1, k}\right\|^2\right|\le \left(1-\lambda_k^2\right)\varepsilon_1, \quad\forall k=1,\cdots,d_1
\end{align}
hold with probability at least $1-{\rm e}^{-c_{2,1}(\varepsilon_1)n}$.
\end{lem}

\begin{proof}
The proof is postponed to Section \ref{subsectionB}.
\end{proof}

Plugging \eqref{eq-vTV2-aTV2} into \eqref{eq-diff1-sum}, we have
\begin{equation}\label{eq-proof-thm-sub-bound1}
\left|\affYS-\left\|{\bf V}_2^{\rm T}\bar{\bf A}_1\right\|_{\rm F}^2\right| \le \varepsilon_1\sum_{k=1}^{d_1}\left(1-\lambda_k^2\right)
\end{equation}
hold with probability at least $1-d_1{\rm e}^{-c_{2,1}(\varepsilon_1)n}$ for any $n>c_{1,1}(\varepsilon_1)d_2$.

{\bf Step 2)} For the second estimation error in the RHS of \eqref{eq-proof-thm-sub-decomp},
we can convert this problem about one subspace $\mathcal{Y}_1$ with dimension $d_1$ into $d_1$ subproblems, each of which is about $1$-dimensional subspace, and then use Lemma \ref{lem-line-sub} to estimate the error.
Denote ${\set X}_{1,k}:=\Span\{{\bf u}_{1,k}\}$, ${\set Y}_{1,k}:=\Span\{{\bf a}_{1,k}\}$, $1\le k\le d_1$.
According to the definition of affinity and Remark \ref{quasi-construct}, we have that the affinity between ${\set X}_{1,k}$ and ${\set X}_2$ is $\|{\bf U}_2^{\rm T}{\bf u}_{1,k}\|=\lambda_k$, and the affinity between ${\set Y}_{1,k}$ and ${\set Y}_2$ is equal to $\left\|{\bf V}_2^{\rm T}\bar{\bf a}_{1,k}\right\|$.
Using Lemma \ref{lem-line-sub},
where ${\set X}_{1,k}$ and ${\set X}_2$ are the original subspaces and ${\set Y}_{1,k}$ and ${\set Y}_2$ are the projected subspaces, we have
\begin{align}\label{aTV2-oaffYk}
\left|\left\|{\bf V}_2^{\rm T}\bar{\bf a}_{1,k}\right\|^2-\overline{\aff}_{{\set Y}_{k}}^2\right|\le\left(1-\lambda_k^2\right)\varepsilon_2,\quad k=1,\cdots,d_1
\end{align}
hold with probability at least $1-d_1{\rm e}^{-c_{2,2}(\varepsilon_2)n}$ for any $n>c_{1,2}(\varepsilon_2)d_2$, where
\begin{equation}\label{eq-proof-thm-sub-single-estimate}
\overline{\aff}_{{\set Y}_{k}}^2:=\lambda_k^2+\frac{d_2}{n}\left(1-\lambda_k^2\right).
\end{equation}

Plugging the definition of $\oaffYS$ in \eqref{thm-sub-aff-change}, \eqref{aTV2-oaffYk}, and \eqref{eq-proof-thm-sub-single-estimate} into the second estimation error, we have
\begin{align}
\left|\left\|{\bf V}_2^{\rm T}\bar{\bf A}_1\right\|_{\rm F}^2-\oaffYS\right|
&= \left|\left\|{\bf V}_2^{\rm T}\bar{\bf A}_1\right\|_{\rm F}^2-\left(\affXS+\frac{d_2}{n}(d_1-\affXS)\right)\right|\nonumber\\
&= \left|\sum_{k= 1}^{d_1} \left( \left\|{\bf V}_2^{\rm T}\bar{\bf a}_{1,k}\right\|^2-\left(\lambda_k^2+\frac{d_2}{n}(1-\lambda_k^2)\right) \right)\right| \nonumber\\
&\le \sum_{k= 1}^{d_1} \left|\left\|{\bf V}_2^{\rm T}\bar{\bf a}_{1,k}\right\|^2-\overline{\rm aff}_{{\set Y}_k}^2\right| \nonumber\\
&\le \varepsilon_2\sum_{k= 1}^{d_1} (1-\lambda_k^2).\label{eq-aTV2-oaffYS}
\end{align}

{\bf Step 3)} Combining \eqref{eq-proof-thm-sub-bound1} and \eqref{eq-aTV2-oaffYS} together into \eqref{eq-proof-thm-sub-decomp}, we have that for any $n>\max\{c_{1,1}(\varepsilon_1),c_{1,2}(\varepsilon_2)\}d_2$,
\begin{align}\label{eq-combining}
\nonumber \left|\affYS-\oaffYS\right|
\le&\left(\varepsilon_1+\varepsilon_2\right)\sum_{k=1}^{d_1}\left(1-\lambda_k^2\right)\nonumber\\
= & \left(\varepsilon_1+\varepsilon_2\right)\left(d_1-\affXS\right)
\end{align}
hold with probability at least $1 - d_1{\rm e}^{-c_{2,1}(\varepsilon_1)n}-d_1{\rm e}^{-c_{2,2}(\varepsilon_2)n}$.
Let $\varepsilon_1=\varepsilon_2=:\varepsilon/2$, according to Remark \ref{exp-summation},
one can easily verify that there exist constants $c_1$, $c_2$ depending only on $\varepsilon$, when $n>c_1(\varepsilon)d_2$,
$$
\left|\affYS-\oaffYS\right|\le\left(d_1-\affXS\right)\varepsilon
$$
holds with probability at least $1-{\rm e}^{-c_2(\varepsilon)n}$.
Then we complete the proof.

\subsection{Proof of Lemma \ref{vTV2-aTV2}}
\label{subsectionB}

In order to improve the readability of the proof, we define intensively all the variables required in advance.
Not that some variables defined before are also summarized here to make this part self-contained.

We use ${\set X}_1$ and ${\set X}_2$ to denote the subspaces before projection, with dimensions $d_1 \le d_2$.
The \emph{principal} orthonormal bases for ${\set X}_1$ and ${\set X}_2$ are denoted as ${\bf U}_1$ and ${\bf U}_2$, respectively.
The $k$th column of ${\bf U}_i$ is denoted as ${\bf u}_{i,k}$, which spans a $1$-dimensional subspace denoted as ${\set X}_{i,k}, k=1,\cdots,d_i, i=1,2$.
In addition, we define ${\bf U}_{i,1:k}$ as the matrix composed of the first $k$ columns of ${\bf U}_i$.
That is ${\bf U}_{i,1:k}=\left[{\bf u}_{i,1},\cdots,{\bf u}_{i,k}\right], 1\le k\le d_i, i=1,2$.
The subspace spanned by the columns of ${\bf U}_{i,1:k}$ is denoted as ${\set X}_{i,1:k}=\Span({\bf U}_{i,1:k})$.

We use ${\set Y}_1$ and ${\set Y}_2$ to denote the subspaces after projection, respectively, from ${\set X}_1$ and ${\set X}_2$ by using a standard Gaussian random matrix ${\bm\Phi}$.
The dimensions of ${\set Y}_1$ and ${\set Y}_2$ stay to be $d_1$ and $d_2$ with probability $1$.
${\bf A}_i={\bm\Phi}{\bf U}_i$ is a basis for ${\set Y}_i$ and its $k$th column is denoted as ${\bf a}_{i,k}$, which spans a 1-dimensional subspace denoted as ${\set Y}_{i,k} = \Span({\bf a}_{i,k})$.
We define ${\bf A}_{i,1:k}=\left[{\bf a}_{i,1},\cdots,{\bf a}_{i,k}\right]$ as the composition of the first $k$ columns of ${\bf A}_i$.
The subspace spanned by the columns in ${\bf A}_{i,1:k}$ is denoted as ${\set Y}_{i,1:k}=\Span({\bf A}_{i,1:k})$.

We use $\bar{\bf A}_1$ and $\bar{\bf A}_2$ to denote the column-normalized result of ${\bf A}_1$ and ${\bf A}_2$, respectively.
${\bf V}_1$ and ${\bf V}_2$ are defined as the orthonormalized result of $\bar{\bf A}_1$ and $\bar{\bf A}_2$, respectively, by using Gram-Schmidt orthogonalization.
As a consequence, ${\bf V}_i$ provides an orthonormal basis for ${\set Y}_i$.
Similarly, ${\bf V}_{i,1:k}$ denotes the matrix composed of the first $k$ columns of ${\bf V}_i$.

Let's start the proof of Lemma \ref{vTV2-aTV2} from the LHS of \eqref{eq-vTV2-aTV2}.
According to the definition of ${\bf V}_2, {\bf v}_{1,k}, \bar{\bf a}_{1,k}$, and Remark \ref{principalangle-affinity-projection},
we have
\begin{align}
	\|{\bf V}_2^{\rm T}{\bf v}_{1,k}\|^2 &= \left\|\Proj_{{\set Y}_2}({\bf v}_{1,k})\right\|^2
	= 1 - \left\|\Proj_{{\set Y}_2^{\perp}}({\bf v}_{1,k})\right\|^2,\\
	\|{\bf V}_2^{\rm T}\bar{\bf a}_{1,k}\|^2 &= \left\|\Proj_{{\set Y}_2}(\bar{\bf a}_{1,k})\right\|^2
	= 1 - \left\|\Proj_{{\set Y}_2^{\perp}}(\bar{\bf a}_{1,k})\right\|^2.
\end{align}
As a consequence, the LHS of \eqref{eq-vTV2-aTV2} is derived as the difference of squared norm of the projection of ${\bf v}_{1,k}$ and $\bar{\bf a}_{1,k}$ onto to the orthogonal complement of ${\set Y}_2$, i.e.
\begin{equation}\label{eq-vTV2-aTV2-derive}
\left|\|{\bf V}_2^{\rm T}{\bf v}_{1,k}\|^2 - \|{\bf V}_2^{\rm T}\bar{\bf a}_{1,k}\|^2\right|
= \left|\left\|\Proj_{{\set Y}_2^{\perp}}({\bf v}_{1,k})\right\|^2 - \left\|\Proj_{{\set Y}_2^{\perp}}(\bar{\bf a}_{1,k})\right\|^2\right|.
\end{equation}

In order to analyze ${\bf v}_{1,k}$,
we take a close look at the Gram-Schmidt orthogonalization process.
We introduce
\begin{equation}\label{eq-define-rk}
\alpha_k:=\|\Proj_{{\set Y}_{1,1:k-1}}(\bar{\bf a}_{1,k})\| = \|{\bf V}_{1,1:k-1}^{\rm T}\bar{\bf a}_{1,k}\|
\end{equation}
as the cosine of the only principal angle between ${\set Y}_{1,k}$ and ${\set Y}_{1,1:k-1}$,
and
\begin{equation}\label{eq-define-bk}
{\bf b}_k := \frac1{\alpha_k}\Proj_{{\set Y}_{1,1:k-1}}(\bar{\bf a}_{1,k})
\end{equation}
as a unit vector along the direction of the projection of ${\bf a}_{1,k}$ onto ${\set Y}_{1,1:k-1}$.
As a consequence, the Gram-Schmidt orthogonalization process is represented by
\begin{align}\label{v1k-new}
\bar{\bf a}_{1,k} &= \Proj_{{\set Y}_{1,1:k-1}}(\bar{\bf a}_{1,k}) + \Proj_{{\set Y}_{1,1:k-1}^\perp}(\bar{\bf a}_{1,k})\nonumber\\
&= \alpha_k {\bf b}_k + \sqrt{1 - \alpha_k^2}{\bf v}_{1, k}.
\end{align}

Then we introduce
\begin{align}
\hat\lambda_k &:= \|\Proj_{{\set Y}_2}(\bar{\bf a}_{1, k})\| = \|{\bf V}_2^{\rm T}\bar{\bf a}_{1,k}\|,\label{eq-defi-barlambdak}\\
\beta_k & := \|\Proj_{{\set Y}_2}({\bf b}_k)\| = \|{\bf V}_2^{\rm T}{\bf b}_k\|\label{eq-defi-betak}
\end{align}
to denote, respectively, the cosine of the only principal angle between ${\set Y}_{1,k}$ and ${\set Y}_2$ and that between $\Span\{{\bf b}_k\}$ and ${\set Y}_2$.
Now projecting both side of \eqref{v1k-new} on the orthogonal complement of ${\set Y}_2$, we have
\begin{align}\label{v1kperp}
\sqrt{1-\hat\lambda_k^2}\bar{\bf a}_{1,k}^{\perp} = \alpha_k\sqrt{1-\beta_k^2}{\bf b}_k^{\perp} + \sqrt{1 - \alpha_k^2}
\Proj_{{\set Y}_2^\perp}({\bf v}_{1, k}),
\end{align}
where $\bar{\bf a}_{1,k}^{\perp}$ and ${\bf b}_k^{\perp}$ denotes, respectively, the unit vectors along $\Proj_{{\set Y}_2^\perp}(\bar{\bf a}_{1,k})$ and $\Proj_{{\set Y}_2^\perp}({\bf b}_{k})$.

Moving the first item in the RHS of \eqref{v1kperp} to the LHS and then taking norm on both sides, we get
\begin{equation}\label{v1kperpnorm}
(1-\alpha_k^2)\|\Proj_{{\set Y}_2^\perp}({\bf v}_{1, k})\|^2=1-\hat\lambda_k^2+\alpha_k^2\left(1-\beta_k^2\right)-2\alpha_k\sqrt{1-\hat\lambda_k^2}\sqrt{1-\beta_k^2}\langle \bar{\bf a}_{1,k}^{\perp}, {\bf b}_k^{\perp} \rangle.
\end{equation}
In addition, the norm of the projection of $\bar{\bf a}_{1,k}$ onto to ${\set Y}_2^\perp$ could be directly represented by using $\hat\lambda_k$ as
\begin{equation}\label{bara1kperpnorm}
\|\Proj_{{\set Y}_2^{\perp}}(\bar{\bf a}_{1,k})\|^2 = 1-\hat\lambda_k^2.
\end{equation}
Inserting both \eqref{v1kperpnorm} and \eqref{bara1kperpnorm} into \eqref{eq-vTV2-aTV2-derive}, we write
\begin{align}
&\left|\|{\bf V}_2^{\rm T}{\bf v}_{1,k}\|^2 - \|{\bf V}_2^{\rm T}\bar{\bf a}_{1,k}\|^2\right|\nonumber\\
= & \left|\frac{1-\hat\lambda_k^2+\alpha_k^2\left(1-\beta_k^2\right)-2\alpha_k\sqrt{1-\hat\lambda_k^2}\sqrt{1-\beta_k^2}\langle \bar{\bf a}_{1,k}^{\perp}, {\bf b}_k^{\perp} \rangle}{1-\alpha_k^2} - (1-\hat\lambda_k^2)  \right|\nonumber\\
\le & \frac{\alpha_k^2\left(1-\hat\lambda_k^2+1-\beta_k^2\right)+2\left| \alpha_k\sqrt{1-\hat\lambda_k^2}\sqrt{1-\beta_k^2}\langle \bar{\bf a}_{1,k}^{\perp}, {\bf b}_k^{\perp} \rangle\right|}{1-\alpha_k^2}.\label{eq-vTV2-aTV2-derive1}
\end{align}
Using the fact that geometric mean is no more than arithmetic mean and $\alpha_k^2<1/3$,
which will be verified soon that $\alpha_k^2$ is a small quantity, we further reshape \eqref{eq-vTV2-aTV2-derive1} as
\begin{equation}
\left|\|{\bf V}_2^{\rm T}{\bf v}_{1, k}\|^2-\left\|{\bf V}_2^{\rm T}\bar{\bf a}_{1,k}\right\|^2 \right|
\le \frac{3}2\left(1-\hat\lambda_k^2+1-\beta_k^2\right)\left( \alpha_k^2 +\left|\alpha_k\langle \bar{\bf a}_{1,k}^{\perp}, {\bf b}_k^{\perp} \rangle\right|\right).\label{eq-vTV2-aTV2-derive2}
\end{equation}

In the following, we will estimate the four quantities of $1-\hat\lambda_k^2$, $\alpha_k^2$, $1-\beta_k^2$, and $\langle \bar{\bf a}_{1,k}^{\perp}, {\bf b}_k^{\perp} \rangle$ separately.

Let us first consider $1-\hat\lambda_k^2$.
Recalling its definition in \eqref{eq-defi-barlambdak}, we have already estimate $\hat\lambda_k^2$ in \eqref{aTV2-oaffYk}.
Inserting \eqref{eq-proof-thm-sub-single-estimate} into \eqref{aTV2-oaffYk},
with probability at least $1-{\rm e}^{-c_{2,1}(\varepsilon_1)n}$,
we have for any $n>c_{1,1}(\varepsilon_1)d_2$
\begin{equation}\label{eq-bound-1-barlambdak2-prepare}
\lambda_k^2+\frac{d_2}{n}\left(1-\lambda_k^2\right) - \hat\lambda_k^2\le\left(1-\lambda_k^2\right)\varepsilon_1,
\end{equation}
Using basic algebra, \eqref{eq-bound-1-barlambdak2-prepare} is reshaped to
\begin{align}\label{eq-1-lambda}
1-\hat\lambda_k^2\le&1-\lambda_k^2-\frac{d_2}{n}\left(1-\lambda_k^2\right)+\left(1-\lambda_k^2\right)\varepsilon_1\nonumber\\
=&\left(1-\lambda_k^2\right)\left(1-\frac{d_2}{n}+\varepsilon_1\right)\nonumber\\
< & \left(1-\lambda_k^2\right)\left(1+\varepsilon_1\right), \quad k=1,\cdots,d_1.
\end{align}

The bound of \eqref{eq-1-lambda} looks direct because $\hat\lambda_k$ denotes the affinity compressed from $\lambda_k$.
According to Lemma \ref{lem-line-sub}, the former can be estimated by the latter.

Second let us check $\alpha_k = \|{\bf V}_{1,1:k-1}^{\rm T}\bar{\bf a}_{1,k}\|$.
Intuitively, because ${\set X}_{1,1:k-1}$ is orthogonal to ${\set X}_{1,k}$, the new subspaces ${\set Y}_{1,1:k-1}$ (projected from ${\set X}_{1,1:k-1}$) and ${\set Y}_{1,k}$ (projected from ${\set X}_{1,k}$) are approximately orthogonal to each other.
Actually, ${\bf V}_{1,1:k-1}$ and $\bar{\bf a}_{1,k}$ satisfy all conditions in Corollary \ref{cor-lem5}.
As a consequence, there exist constants $c_{1,2}(\varepsilon_2)$, $c_{2,2}(\varepsilon_2)$, such that for any $\varepsilon_2<\frac13$ and $n>c_{1,2}(\varepsilon_2)d_1>c_{1,2}(\varepsilon_2)(k-1)$, we have $\alpha_k^2<\varepsilon_2$ hold with probability at least $1-{\rm e}^{-c_{2,2}(\varepsilon_2)n}$.

Next we consider $1-{\beta}_k^2$, which is also bounded by $1-{\lambda}_k^2$,
Notice that ${\bf b}_k$ lies in ${\set Y}_{1,1:k-1}\subset{\set Y}_{1,1:k}$ and $\beta_k$ is the norm of the projection of ${\bf b}_k$ onto ${\set Y}_2$.
Because the minimum of the norm of the projection of a unit vector in ${\set Y}_{1,1:k}$ onto ${\set Y}_2$ approximates $\lambda_k$, $1-\beta_k^2$ should be very close to $1-\lambda_k^2$.
The difference between them can be bounded by the following lemma.

\begin{lem}\label{lem-property8}
For any $n>c_{1,3}(\varepsilon_3)d_2$, we have
\begin{align}\label{eq7}
1-\beta_k^2\le \left(1-\lambda_k^2\right)\left(1+\varepsilon_3\right)
\end{align}
holds with probability at least $1-{\rm e}^{-c_{2,3}(\varepsilon_3)n}$.
\end{lem}

\begin{proof}
The proof is postponed to Appendix \ref{proof-lem-property8}.
\end{proof}

Finally, as for the last term to be estimated, $\langle \bar{\bf a}_{1,k}^{\perp}, {\bf b}_k^{\perp} \rangle$ is proved to be a small quantity in Lemma \ref{lem-property7}.
Intuitively, $\bar{\bf a}_{1,k}^{\perp}$ and ${\bf b}_k^{\perp}$ is unit projections of $\bar{\bf a}_{1,k}\in{\set Y}_{1,k}$ and ${\bf b}_k\in{\set Y}_{1,1:k-1}$, respectively, onto ${\set Y}_2^{\perp}$.
Consequently, the inner product between $\bar{\bf a}_{1,k}^{\perp}$ and ${\bf b}_k^{\perp}$ should be very small if ${\set Y}_{1,k}$ and ${\set Y}_{1,1:k-1}$, which are independent with each other, are both independent with ${\set Y}_2^{\perp}$.

\begin{lem}\label{lem-property7}
There exist constants $c_{1,4}(\varepsilon_4)$, $c_{2,4}(\varepsilon_4)$, such that for any $n>c_{1,4}(\varepsilon_4)d_2$, we have
$
\left|\langle \bar{\bf a}_{1,k}^{\perp}, {\bf b}_k^{\perp} \rangle\right|^2 < \varepsilon_4
$
holds with probability at least $1-{\rm e}^{-c_{2,4}(\varepsilon_4)n}$.
\end{lem}

\begin{proof}
The proof is postponed to Appendix \ref{proof-lem-property7}.
\end{proof}

Now, we are ready to complete the proof by using the concentration properties derived above.
Plugging \eqref{eq-1-lambda}, \eqref{eq7}, $\alpha_k^2<\varepsilon_2$, and $\left|\langle \bar{\bf a}_{1,k}^{\perp}, {\bf b}_k^{\perp} \rangle\right|^2 < \varepsilon_4$ into \eqref{eq-vTV2-aTV2-derive2},
we have for any $n>\max\{c_{1,l}(\varepsilon_l)\}d_2$, $l=1,2,3,4$,
\begin{align*}
\left|\|{\bf V}_2^{\rm T}{\bf v}_{1, k}\|^2-\left\|{\bf V}_2^{\rm T}\bar{\bf a}_{1,k}\right\|^2 \right|
\le&\frac{3}2\left(1-\lambda_k^2\right)\left(2+\varepsilon_1+\varepsilon_3\right)\left(\varepsilon_2+\sqrt{\varepsilon_2\varepsilon_4}\right)
\end{align*}
hold with probability at least
$
1-\sum_{l=1}^4{\rm e}^{-c_{2,l}(\varepsilon_l)n}.
$
Let $\varepsilon_1<1$ and $\varepsilon_3<1$, $\varepsilon_2=\varepsilon_4=:\varepsilon/12$, then we have
\begin{align}\label{vTV-aTV-bound}
\left|\|{\bf V}_2^{\rm T}{\bf v}_{1, k}\|^2-\left\|{\bf V}_2^{\rm T}\bar{\bf a}_{1,k}\right\|^2 \right|\le\left(1-\lambda_k^2\right)\varepsilon.
\end{align}
According to Remark \ref{exp-summation}, we claim that there exist constants $c_1$, $c_2$, such that for any $n>c_1(\varepsilon)d_2$, \eqref{vTV-aTV-bound} holds with probability at least $1-{\rm e}^{-c_2(\varepsilon)n}$.

\section{Related Works}
\label{Relatedworks}

Our earlier results on the RIP of subspaces in \cite{li2017restricted} are cited below.

\begin{thm}\label{thm-rip-previous}
Suppose $\set X_1, \set X_2\subset \mathbb{R}^N$ are two subspaces with dimension $d_1 \le d_2$, respectively.
If $\set X_1$ and $\set X_2$ are projected into $\mathbb R^n$ by a Gaussian random matrix ${\bm\Phi} \in\mathbb{R}^{n\times N}$, $\set{X}_k \stackrel{\bm \Phi}{\longrightarrow} \set{Y}_k, k=1,2$, then we have
$$
    (1-\varepsilon)\dXS \le \dYS  \le (1+\varepsilon)\dXS,
$$
with probability at least
$$
	1 - \frac{4d_1}{(\varepsilon-{d_2}/{n})^2n},
$$
when $n$ is large enough.
\end{thm}

\begin{thm}\label{thm-rip-set-previous}
For any set composed by $L$ subspaces ${\set X}_1, \cdots, {\set X}_L \in \mathbb{R}^N$ of dimension no more than $d$,
if they are projected into $\mathbb R^n$ by a Gaussian random matrix ${\bm\Phi} \in\mathbb{R}^{n\times N}$,
$
\set{X}_k \stackrel{\bm \Phi}{\longrightarrow} \set{Y}_k, k=1,\cdots,L,
$
and $d\ll n < N$, then we have
$$
    (1-\varepsilon)D^2({\set X}_i,{\set X}_j) \le D^2({\set Y}_i,{\set Y}_j) \le (1+\varepsilon)D^2({\set X}_i,{\set X}_j), \quad \forall i, j
$$
with probability at least
$$
	1 - \frac{2dL(L-1)}{(\varepsilon-d/{n})^2n},
$$
when $n$ is large enough.
\end{thm}

Compared with the our previous results, this paper has the following two main improvements.
Firstly, because we use more advanced random matrix theories and deal with the error more skillfully, 
the probability bound $1 - {\rm e}^{-\mathcal{O}(n)}$ derived in this paper is much tighter than 
the $1 - \mathcal{O}(1/n)$ in the previous work, where we used Chebyshev inequality.
Such improvement provides a more accurate law of magnitude of the dimensions
in this random projection problem,
and the improved probability bound is optimum, if one compares it with the analogical conclusions 
in the theory of Compressed Sensing.
Secondly, Theorem \ref{thm-rip-set-previous} requires $n \gg d$,
but it does not specify how large $n$ should be or the connection between $\varepsilon$, $d$, $L$, and the lower bound of $n$. 
In comparison, Theorem \ref{thm-rip} in this paper rigorously clarifies that the conclusion will hold
as long as $n$ is larger than $c_1(\varepsilon)\max\{d, \ln L\}$.

\section{Conclusion}
\label{secConclusion}

In this paper, we utilize the random matrix theory to rigorously prove the RIP of Gaussian random compressions 
for low-dimensional subspaces.
Mathematically, we demonstrate that as long as the dimension after compression $n$ is larger than 
$c_1(\varepsilon)\max\{d, \ln L\}$, with probability no less than $1 - {\rm e}^{-c_2(\varepsilon)n}$,
the distance between any two subspaces after compression remains almost unchanged.
The probability bound $1 - {\rm e}^{-\mathcal{O}(n)}$ is optimum in the asymptotic sense, 
in comparison with the analogical optimum theoretical result of RIP in Compressed Sensing.
Our work can provide a solid theoretical foundation for Compressed Subspace Clustering and other low-dimensional subspace related problems.

\section{Appendix}
\label{secAppendix}

\subsection{Proof of Lemma \ref{exp-formal}}
\label{proof-exp-formal}

We first prove the special case that $K = 1$, i.e., $f = a{\rm e}^{-g}$ for short.
According to \eqref{exp-formal-condition1}, \eqref{exp-formal-condition2}, and the definition of limitation, there exist constants $n_0$ and $c_1 > 0$ depending only on $\varepsilon$. When $n>n_0$, $\tau<c_1$, we have $\frac{g}n>h - \frac{h-b}{3}$, and $\frac{\ln a}{n} < b + \frac{h-b}{3}$.
Let $c_2:=\frac{h-b}{3} > 0$ depending only on $\varepsilon$, we can have
\begin{align*}
f &= a{\rm e}^{-g} = {\rm e}^{-(g-\ln a)} = \exp\left(-n\left(\frac{g}{n} - \frac{\ln a}{n}\right)\right)\\
&\le \exp\left(-n\left(h - \frac{h-b}{3} - b - \frac{h-b}{3}\right)\right) \\
&= \exp\left(-\frac{h-b}{3}n\right) = {\rm e}^{-c_2n}.
\end{align*}

Now we consider the general case of arbitrary $K$.
According to the above analysis, we have that, for each term of $f$, there exist constants $n_{0,k}, c_{1,k} > 0$, and $c_{2,k} > 0$ depending only on $\varepsilon$. When $n>n_{0,k}$, $\tau<c_{1,k}$, it satisfies that
$
a_k{\rm e}^{-g_k} < {\rm e}^{-c_{2,k} n}.
$
Let
$
n_0 :=\max_k n_{0,k},\quad c_1 :=\min_k c_{1,k} > 0,\quad  c_2 :=\min_k c_{2,k} > 0.
$
Then when $n>n_0$, $\tau<c_1$, we have that
$$
f = \frac{1}{K}\sum_{k = 1}^K a_k{\rm e}^{-g_k}<\frac{1}{K}\sum_{k = 1}^K{\rm e}^{-nc_{2,k}}\le {\rm e}^{-c_2n},
$$
and complete the proof.

\subsection{Proof of Lemma \ref{lem2}}
\label{proof-lem2}

Regarding $\sqrt{n}\va$ as a matrix belonging to $\mathbb{R}^{n\times 1}$,
and using Lemma \ref{T},
we have that with probability at least $1 - 2{\rm e}^{-\frac{t^2}2}$,
\begin{align*}
\sqrt{n}-1-t \leq s_{\min}(\sqrt{n}\va) =\sqrt{n}\|\va\| = s_{\max}(\sqrt{n}\va) \leq \sqrt{n}+1+t.
\end{align*}
Taking square and subtracting $n$ from both sides, we have
$$
-\left(2\sqrt{n}(1+t)+\left(1+t\right)^2\right)\le-\left(2\sqrt{n}(1+t)-\left(1+t\right)^2\right) \leq n\|\va\|^2-n \leq 2\sqrt{n}(1+t)+\left(1+t\right)^2,
$$
with probability at least $1 - 2{\rm e}^{-\frac{t^2}2}$.

By choosing $\varepsilon$ satisfying $n\varepsilon = 2\sqrt{n}(1+t) + \left(1+t\right)^2$, we can get
\begin{align}\label{eq32}
t=\sqrt{n}\left(\sqrt{1+\varepsilon}-1-{1}/{\sqrt{n}}\right).
\end{align}
When $n>\left(\frac{1}{\sqrt{\varepsilon+1}-1}\right)^2=:n_{0,1}$, we have $t>0$.
Substituting this equation into the expression of probability, we have
\begin{align}\label{eq33}
\mathbb{P}\left(\left| \left\|{\bf a}\right\|^2-1 \right|>\varepsilon \right)<2\exp\left(-n\left(\sqrt{1+\varepsilon}-1-{1}/{\sqrt{n}}\right)^2/2\right).
\end{align}
According to Lemma \ref{exp-formal}, there exist constants $n_{0,2}$ and $c$ dependent on $\varepsilon$, such that the ${\rm RHS}$ of \eqref{eq33} is smaller than ${\rm e}^{-cn}$.
Taking $n_0=\max\{n_{0,1},n_{0,2}\}$, we complete the proof.

\subsection{Proof of Corollary \ref{cor-lem2}}
\label{proof-cor-lem2}

\begin{proof}
Notice that $\sqrt{\frac{n}d}{\bf V}^{\rm T}{\bf a}\in\mathbb{R}^{d}$ is a standard Gaussian random vector.
As a consequence, according to \eqref{eq33} in the proof of Lemma \ref{lem2}, by replacing $\varepsilon$ with $\frac{n\varepsilon}{d}$, we have
\begin{align}\label{eq34}
\mathbb{P}\left( \left|\left\|\sqrt{\frac{n}d}{\bf V}^{\rm T}{\bf a}\right\|^2 - 1\right|> \frac{n\varepsilon}{d}\right) < 2\exp\left(-d\left(\sqrt{1+\frac{n}{d}\varepsilon}-1-{1}/{\sqrt{d}}\right)^2/2\right),
\end{align}
where $d>\left(\frac{1}{\sqrt{\frac{n}{d}\varepsilon+1}-1}\right)^2$ is required.
In order to satisfy this requirement, i.e., $\left(\frac{1}{\sqrt{\frac{n}{d}\varepsilon+1}-1}\right)^2<1\le d$,
we need $n>\frac{3d}{\varepsilon}=:c_{1,1}d$.
According to Lemma \ref{exp-formal}, there exist constants $c_{1,2}$, $c_2$ dependent on $\varepsilon$, such that the $\rm RHS$ of \eqref{eq34} is smaller than ${\rm e}^{-c_2n}$.
Taking $c_1:=\max\{c_{1,1}, c_{1,2}\}$ and dividing both sides of the expression in $\mathbb{P}(\cdot)$ in \eqref{eq34} by $n/d$, we complete the proof.
\end{proof}

\subsection{Proof of Corollary \ref{unit-singular-value}}
\label{proof-unit-singular-value}

In order to bound $s_{\min}^2\left(\bar{\bf A}\right)$, noticing that
\begin{align}\label{eq23}
s_{\min}^2\left(\bar{\bf A}\right)\ge\frac{s_{\min}^2\left({\bf A}\right)}{\max\limits_i\left\|{\bf a}_i\right\|^2},
\end{align}
we may turn to estimate $s_{\min}^2\left({\bf A}\right)$ and $\max_i\left\|{\bf a}_i\right\|^2$ separately.

We begin from estimating $s_{\min}^2\left({\bf A}\right)$.
According to Lemma \ref{T}, with probability at least $1-{\rm e}^{-t^2/2}$, we have
\begin{align}\label{smin-A}
s_{\min}^2\left({\bf A}\right)\ge \frac{1}{n}\left(\sqrt{n}-\sqrt{k}-t\right)^2=\left(1-\sqrt{{k}/{n}}-{t}/{\sqrt{n}}\right)^2.
\end{align}
Let $1-\varepsilon_1$ be the ${\rm RHS}$ of \eqref{smin-A} then we have
\begin{align}\label{eq24}
t=\sqrt{n}\left(1-\sqrt{1-\varepsilon_1}-\sqrt{{k}/{n}}\right).
\end{align}
When $n>\frac{k}{\left(1-\sqrt{1-\varepsilon_1}\right)^2}=:\hat{c}_{0,1}k$, we have $t>0$.
Plugging \eqref{eq24} into ${\rm e}^{-t^2/2}$, we can get the probability that \eqref{eq24} violates as
$
\exp\left(-n\left(1-\sqrt{{k}/{n}}-\sqrt{1-\varepsilon_1}\right)^2/2\right).
$
According to Lemma \ref{exp-formal}, the above probability can be bounded by ${\rm e}^{-\hat{c}_{2,1}(\varepsilon_1)n}$ for $n>\hat{c}_{1,1}k$.
Then we have
\begin{align}\label{smin-bar-A-pre}
s_{\min}^2\left({\bf A}\right)\ge 1-\varepsilon_1
\end{align}
hold for $n>\max\{\hat{c}_{1,1},\hat{c}_{0,1}\}k=:\hat{c}_3k$ with probability at least $1-{\rm e}^{-\hat{c}_{2,1}(\varepsilon_1)n}$.

Next we estimate $\max_i\left\|{\bf a}_i\right\|^2$.
According to Lemma \ref{lem2}, with probability at least $1-k{\rm e}^{-\hat{c}_{2,2}(\varepsilon_2)n}$, for $n>n_{0,1}$, we have
\begin{align}\label{norm-max-a1k}
\max_i\left\|{\bf a}_{i}\right\|^2\le 1+\varepsilon_2.
\end{align}
Plugging \eqref{smin-bar-A-pre} and \eqref{norm-max-a1k} into \eqref{eq23}, we have
$$
s_{\min}^2\left(\bar{\bf A}\right)\ge\frac{1-\varepsilon_1}{1+\varepsilon_2}=1-\frac{\varepsilon_1+\varepsilon_2}{1+\varepsilon_2}.
$$
Let $\varepsilon_1=\varepsilon_2=:\varepsilon/2$, with probability at least $1-{\rm e}^{-\hat{c}_{2,1}(\varepsilon/2)n}-k{\rm e}^{-\hat{c}_{2,2}(\varepsilon/2)n}$, we have
\begin{align}\label{smax-A-bar2}
s_{\min}^2\left(\bar{\bf A}\right)\ge 1-\left({\varepsilon}/{2}+{\varepsilon}/{2}\right)=1-\varepsilon.
\end{align}
Take $c_{1,1}:=\max\{\hat{c}_{3}, n_{0,1}\}$ and we prove the first part of this corollary.

In order to bound $s_{\max}^2\left(\bar{\bf A}\right)$, following the same approach, we could derive step by step the counterparts of \eqref{eq23}, \eqref{smin-bar-A-pre}, and \eqref{norm-max-a1k}, respectively, as
\begin{align}\label{eq25}
s_{\max}^2\left(\bar{\bf A}\right)\le\frac{s_{\max}^2\left({\bf A}\right)}{\min\limits_i\left\|{\bf a}_i\right\|^2},
\end{align}
\begin{align}\label{smax-bar-A-pre}
s_{\max}^2\left({\bf A}\right)\le 1+\varepsilon_1
\end{align}
for $n>\hat{c}_4k$ with probability at least
\begin{align*}
1-\exp\left(-n\left(1+\sqrt{{k}/{n}}-\sqrt{1+\varepsilon_1}\right)^2/2\right)>1-{\rm e}^{-\hat{c}_{2,3}n},
\end{align*}
and
\begin{align}\label{norm-min-a1k}
\min_i\left\|{\bf a}_{i}\right\|^2\ge 1-\varepsilon_2
\end{align}
for $n>n_{0,2}$ with probability at least $1-k{\rm e}^{-\hat{c}_{2,2}(\varepsilon_2)n}-{\rm e}^{-\hat{c}_{2,3}(\varepsilon_1)n}$.
Then we have
\begin{align}\label{smax-A-bar1}
s_{\max}^2\left(\bar{\bf A}\right)\le\frac{1+\varepsilon_1}{1-\varepsilon_2}=1+\frac{\varepsilon_1+\varepsilon_2}{1-\varepsilon_2}.
\end{align}
Similarly reshaping \eqref{smax-A-bar1} and letting $\varepsilon_2\le 1/2$, $\varepsilon_1=\varepsilon_2=:\varepsilon/4$, taking $c_{1,2}:=\max\{\hat{c}_4,n_{0,2}\}$, we prove the second part of the corollary.

\subsection{Proof of Lemma \ref{lem5}}
\label{proof-lem5}

Using the orthogonality between ${\bf u}_1$ and ${\bf U}_2$, $\Span({\bf A}_2)$ is independent with ${\bf a}_1$.
As an orthonormal basis of such subspace, ${\bf V}_2$ is also independent with ${\bf a}_1$.
Then, according to Definition \ref{defi-guassian}, ${\bf a}_1$ conditioned on ${\bf V}_2$ is still a standard Gaussian random vector.
As a consequence, $\sqrt{n}{\bf V}_2^{\rm T}{\bf a}_1\in\mathbb{R}^{d\times 1}$, the entries of which are independent standard Gaussian random variables, satisfies the condition in Lemma \ref{T}.
With probability no more than ${\rm e}^{-\frac{t^2}2}$, we have
\begin{align}\label{eq12}
\left\|\sqrt{n}{\bf V}_2^{\rm T}{\bf a}_1\right\|^2=s_{\max}^2\left(\sqrt{n}{\bf V}_2^{\rm T}{\bf a}_1\right)\ge \left(\sqrt{d} + 1 + t\right)^2.
\end{align}
Let $\varepsilon := {\left(\sqrt{d} + 1 + t\right)^2}/{n}$, we can get
\begin{align}\label{eq13}
t=\sqrt{n\varepsilon}-\sqrt{d}-1.
\end{align}
When $n>\frac{4d}{\varepsilon}=:c_{1,1}d$, we have $t>\sqrt{d}-1\ge 0$.
Plugging \eqref{eq13} into ${\rm e}^{-\frac{t^2}2}$, the probability of \eqref{eq12} holding is at least
$
2\exp\left(-\left(\sqrt{n\varepsilon}-\sqrt{d}-1\right)^2/2\right).
$
According to Lemma \ref{exp-formal}, there exist constants $c_{1,2}$, $c_2$, such that when $n>c_{1,2}d$, this probability is smaller than ${\rm e}^{-c_2n}$.
Taking $c_1:=\max\{c_{1,1},c_{1,2}\}$ and dividing both sides of \eqref{eq12} by $n$, we conclude the lemma.

\subsection{Proof of Corollary \ref{cor-lem5}}
\label{proof-cor-lem5}

According to the definition of $\bar{\bf a}_1$ and basic probability, we have
  \begin{align}\label{eqp1}
  \mathbb{P}\left( \left\|{\bf V}_{2}^{\rm T}\bar{\bf a}_{1}\right\|^2 >\varepsilon\right)
  = & \mathbb{P}\left( \frac{\left\| {\bf V}_{2}^{\rm T}{\bf a}_{1}\right\|^2 }{\left\|\va_{1}\right\|^2} >\varepsilon\right)\nonumber\\
  =& 1 - \mathbb{P}\left( \frac{\left\| {\bf V}_{2}^{\rm T}{\bf a}_{1}\right\|^2 }{\left\|\va_{1}\right\|^2} <\varepsilon\right)\nonumber\\
  \le & 1 - \mathbb{P}\left( \Vert \va_{1}\Vert^2 > 1- \varepsilon \ {\rm and} \  \left\|{\bf V}_{2}^{\rm T}\va_{1}\right\|^2<\varepsilon\left(1-\varepsilon\right) \right)\nonumber\\
  = & \mathbb{P}\left( \Vert \va_{1}\Vert^2 < 1- \varepsilon \ {\rm or} \  \left\|{\bf V}_{2}^{\rm T}\va_{1}\right\|^2>\varepsilon\left(1-\varepsilon\right) \right)\nonumber\\
  \le &\mathbb{P}\left(\Vert \va_{1}\Vert^2 < 1- \varepsilon\right) + \mathbb{P}\left(\left\|{\bf V}_{2}^{\rm T}\va_{1}\right\|^2>\varepsilon\left(1-\varepsilon\right) \right).
\end{align}
Now we may estimate the two items in the RHS of \eqref{eqp1}, separately.
By using Lemma \ref{lem2}, for $n>n_0$, we have
\begin{align}\label{eqp2}
\mathbb{P}(\Vert \va_{1}\Vert^2 < 1- \varepsilon)<{\rm e}^{-c_{2,1}(\varepsilon)n}.
\end{align}
By using Lemma \ref{lem5}, for $\varepsilon< \frac{1}{3}$ and $n>c_{1,2}d$, we have
\begin{align}\label{eqp3}
\mathbb{P}\left(\left\| {\bf V}_{2}^{\rm T}\va_{1}\right\|^2>\varepsilon(1-\varepsilon)\right)< {\rm e}^{-c_{2,2}(2\varepsilon/3)n}.
\end{align}
Plugging \eqref{eqp2} and \eqref{eqp3} into \eqref{eqp1}, we readily get
$$
\mathbb{P}\left( \left\|{\bf V}_{2}^{\rm T}\bar{\bf a}_{1}\right\|^2 >\varepsilon\right) < {\rm e}^{-c_{2,1}(\varepsilon)n}+{\rm e}^{-c_{2,2}(2\varepsilon/3)n}.
$$
According to Remark \ref{exp-summation}, we claim that there exist constants $c_1$, $c_2$, such that for any $n>c_1(\varepsilon)d$, we have $\left\|{\bf V}_{2}^{\rm T}\bar{\bf a}_{1}\right\|^2 >\varepsilon$ with probability no more than ${\rm e}^{-c_2(\varepsilon)n}$.

\subsection{Proof of Remark \ref{rem-p2epsilon1new}}
\label{proof-rem-p2epsilon1new}

Replacing $n$ with $n-d_0$ in \eqref{eq-p2epsilon1new}, we can get
\begin{align}\label{temp}
\mathbb{P}\left(\left\|{\bf V}_2^{\rm T}\bar{\bf a}_1\right\|^2>\varepsilon\right) \le {\rm e}^{-\hat{c}_2(\varepsilon)\left(n-d_0\right)}.
\end{align}
We only need to prove that there exist constants $c_1$, $c_2$, when $n>c_1\max\{d,d_0\}$, we have $\mathbb{P}\left(\left\|{\bf V}_2^{\rm T}\bar{\bf a}_1\right\|^2>\varepsilon\right)\le{\rm e}^{-\hat{c}_2(\varepsilon)\left(n-d_0\right)}\le {\rm e}^{-c_2n}$.
Let $\tau:=\frac{d_0}{n}$, according to Lemma \ref{exp-formal}, we have
\begin{align*}
h:&=\lim_{\tau\to 0}\lim_{n\to\infty}\frac{\hat{c}_2(\varepsilon)\left(n-d_0\right)}{n}=\lim_{\tau\to 0}\lim_{n\to\infty}\hat{c}_2\left(1-\tau\right)=\hat{c}_2>0,\\
b:&=\lim_{n\to\infty}\frac{\ln 1}{n}=0<h.
\end{align*}
Then when $n>n_0$, $\tau=\frac{d_0}{n}\le \tau_0$, there exists constant $c_2$, such that ${\rm e}^{-\hat{c}_2(\varepsilon)\left(n-d_0\right)}\le {\rm e}^{-c_2n}$.
By choosing $c_1:=\max\{n_0,\frac{1}{\tau_0},\hat{c}_1\}$, when $n>c_1d_0$, we have $n>n_0$, $\tau=\frac{d_0}{n}\le \tau_0$.
The condition of \eqref{temp} holding, that is $n>\hat{c}_1d$, is also satisfied.

\subsection{Proof of Lemma \ref{lem-property8}}
\label{proof-lem-property8}

Using the definition of $\beta_k$ in \eqref{eq-defi-betak} and the fact that ${\set Y}_{1,1:k-1} \subset {\set Y}_{1,1:k}$, we have
\begin{equation}\label{V2bnorm2}
\beta_k^2 = \left\|{\bf V}_2^{\rm T}{\bf b}_k\right\|^2
= \min_{{{\left\|\bf b\right\|=1}\atop {{\bf b}\in {\set Y}_{1,1:k-1}}}}\|{\bf V}_2^{\rm T}{\bf b}\|^2
\ge \min_{{{\left\|\bf b\right\|=1}\atop {{\bf b}\in {\set Y}_{1,1:k}}}}\|{\bf V}_2^{\rm T}{\bf b}\|^2.
\end{equation}
Removing both side of \eqref{V2bnorm2} from one, we write
\begin{equation}\label{V2bnorm2max}
1 - \beta_k^2 \le 1 - \min_{{{\left\|\bf b\right\|=1}\atop {{\bf b}\in {\set Y}_{1,1:k}}}}\|{\bf V}_2^{\rm T}{\bf b}\|^2
= 1 - \min_{{{\left\|\bf b\right\|=1}\atop {{\bf b}\in {\set Y}_{1,1:k}}}}\|\Proj_{{\set Y}_2}({\bf b})\|^2
= \max_{{{\left\|\bf b\right\|=1}\atop {{\bf b}\in {\set Y}_{1,1:k}}}}\|\Proj_{{\set Y}_2^\perp}({\bf b})\|^2.
\end{equation}
Then we will loose the condition and rewrite the expression of this maximization problem step by step, and finally convert it to a problem about the extreme singular value of random matrix.

For any vector ${\bf b}$ in ${\set Y}_{1,1:k}$, it can be spanned by the columns of $\bar{\bf A}_{1,1:k}$ as
\begin{equation}\label{eq-b2x}
{\bf b} = \bar{\bf A}_{1,1:k}{\bf x} = \sum_{j=1}^k x_j\bar{\bf a}_{1,j},
\end{equation}
where ${\bf x} = \left[x_1, \cdots, x_k\right]^{\rm T}$ denotes the weight vector.
Consequently, the condition of $\|{\bf b}\|=1$ can be loosen to the condition on ${\bf x}$, i.e.,
\begin{align}\label{upper-x}
\left\|{\bf x}\right\|^2 \le
\frac{\left\|\bar{{\bf A}}_{1,1:k}{\bf x}\right\|^2}{s_{\min}^2\left(\bar{{\bf A}}_{1,1:k}\right)}=
\frac{\left\|{\bf b}\right\|^2}{s_{\min}^2\left(\bar{{\bf A}}_{1,1:k}\right)}= \frac{1}{s_{\min}^2\left(\bar{{\bf A}}_{1,1:k}\right)}=:x_{\rm u}.
\end{align}
Inserting \eqref{eq-b2x} and \eqref{upper-x} in \eqref{V2bnorm2max}, we have
\begin{align}\label{eq-1-betak2-bound}
1 - \beta_k^2 &\le \max_{\left\|\bf x\right\|^2\le x_{\rm u}}\left\|\Proj_{{\set Y}_2^\perp}\left(\sum_{j=1}^k x_j\bar{\bf a}_{1,j}\right)\right\|^2\nonumber\\
&= \max_{\left\|\bf x\right\|^2\le x_{\rm u}}\left\| \sum_{j=1}^k x_j\Proj_{{\set Y}_2^\perp}\left(\bar{\bf a}_{1,j}\right)\right\|^2\nonumber\\
&= \max_{\left\|\bf x\right\|^2\le x_{\rm u}}\left\| \sum_{j=1}^k x_j\sqrt{1-\hat\lambda_j^2}\bar{\bf a}_{1,j}^{\perp}\right\|^2,
\end{align}
where $\hat\lambda_{j}$ is defined in \eqref{eq-defi-barlambdak}.

According to \eqref{eq-1-lambda} and the decreasing order of $\lambda_1\ge\cdots\ge\lambda_{d_1}$, we have
\begin{equation}\label{eq-boundbarlambdaj}
1-\hat\lambda_j^2 \le (1-{\lambda}_j^2)(1+\varepsilon)
\le (1-{\lambda}_k^2)(1+\varepsilon), \quad \forall j=1,\cdots,k<d_1,
\end{equation}
hold with probability at least $1-{\rm e}^{-c_{2,1}(\varepsilon_1)n}$ for any $n>c_{1,1}(\varepsilon_1)d_2$.
Inserting \eqref{eq-boundbarlambdaj} in \eqref{eq-1-betak2-bound}, we have
\begin{align}\label{eq-1-betak2-bound1}
1 - \beta_k^2 &\le  (1-{\lambda}_k^2)(1+\varepsilon_1) \max_{\left\|\bf x\right\|^2\le x_{\rm u}}\left\| \sum_{j=1}^k x_j\bar{\bf a}_{1,j}^{\perp}\right\|^2\nonumber\\
&\le (1-{\lambda}_k^2)(1+\varepsilon_1) s_{\max}^2\left(\bar{\bf A}_{1,1:k}^{\perp}\right)x_{\rm u}\nonumber\\
& =(1-{\lambda}_k^2)(1+\varepsilon_1)\frac{s_{\max}^2\left(\bar{\bf A}_{1,1:k}^{\perp}\right)}{s_{\min}^2\left(\bar{\bf A}_{1,1:k}\right)},
\end{align}
where
$
\bar{\bf A}_{1,1:k}^{\perp} = \left[\bar{\bf a}_{1,1}^{\perp},\cdots,\bar{\bf a}_{1,k}^{\perp} \right].
$

Now we need to bound the denominator and numerator in the RHS of \eqref{eq-1-betak2-bound1}, separately.
As to the denominator, according to Corollary \ref{unit-singular-value}, we have for $n>c_{1,2}(\varepsilon_2)d_1$,
\begin{align}
\mathbb{P}\left(s_{\min}^2\left(\bar{\bf A}_{1,1:k}\right)>1-\varepsilon_2\right) >1-{\rm e}^{-c_{2,2}(\varepsilon_2)n}.\label{eq-bound-denominator}
\end{align}

As for estimating the numerator, since that ${\bf A}_{1,1:k}$ is correlated with ${\set Y}_2$,
we can not directly apply the available lemmas about the concentration inequalities of independent Gaussian random matrix.
However, by using the following techniques, we could manage to convert the problem of estimating $s_{\max}^2\left(\bar{\bf A}_{1,1:k}^{\perp}\right)$ to a problem about the singular value of a normalized random matrix satisfying the independence condition.

\begin{rem}\label{quasi-basis-decompose}
Recalling Remark \ref{quasi-construct} about the characteristics of \emph{principal} orthonormal bases ${\bf U}_1$, ${\bf U}_2$ for subspaces ${\set X}_1, {\set X}_2$, and following the decomposition way in the proof of Lemma \ref{lem-line-sub},
we can decompose each column of ${\bf U}_1$ as the projections onto ${\set X}_2$ and its orthogonal complement and get
\begin{align}\label{U1U2U0}
{\bf U}_1 = {\bf U}_{2} {\bm\Lambda} + {\bf U}_0{\bm\Lambda}^\perp,
\end{align}
where
$$
 {\bm\Lambda}^\perp = \left[ \begin{array}{ccc}
     \sqrt{1- {\lambda}_1^2} & & \\ & \ddots& \\ & & \sqrt{1- {\lambda}_{d_1}^2} \end{array} \right],
$$
and $\Span({\bf U}_0)\in\mathbb{R}^{N\times d_1}$ is a subspace of $\mathcal{X}_2^{\perp}$,
that is ${\bf U}_2^{\rm T}{\bf U}_0={\bf 0}$.
\end{rem}

After random projection, the decomposition in Remark \ref{quasi-basis-decompose} changes to
\begin{align}\label{A1A2A0}
{\bf A}_1 = {\bf A}_{2} {\bm\Lambda} + {\bf A}_0{\bm\Lambda}^\perp,
\end{align}
where ${\bf A}_0={\bm\Phi}{\bf U_0}$.
Projecting both side of \eqref{A1A2A0} onto the orthogonal complement of ${\set Y}_2$, we have
$$
\Proj_{{\set Y}_2^\perp}({\bf A}_1) = \Proj_{{\set Y}_2^\perp}({\bf A}_0){\bm\Lambda}^\perp,
$$
which means that the normalized column of $\Proj_{{\set Y}_2^\perp}({\bf A}_1)$, i.e., $\bar{\bf a}_{1,k}^\perp$ defined in \eqref{v1kperp} are exactly identical to the normalized column of $\Proj_{{\set Y}_2^\perp}({\bf A}_0)$, which is denoted as $\bar{\bf a}_{0,k}^\perp, k=1,\cdots,d_1$.
That is
\begin{align}\label{A1-A0}
	\bar{\bf A}_{1:1:k}^\perp = \left[\bar{\bf a}_{0,1}^\perp,\cdots,\bar{\bf a}_{0,k}^\perp\right] =: \bar{\bf A}_{0,1:k}^{\perp}.
\end{align}

Considering its property of isotropy,
a Gaussian random vector remains Gaussian distribution when it is projected to an independent subspace.
This is demonstrated in Remark \ref{Gaussian-projection}.

\begin{rem}\label{Gaussian-projection}
Let ${\bf A}_1\in\mathbb{R}^{n\times d_1}$ and ${\bf A}_2\in\mathbb{R}^{n\times d_2}, d_1\le d_2$ be two Gaussian random matrices.
We denote ${\bf V}_2$ as an orthonormal basis of $\Span({\bf A}_2)$.
The projection of ${\bf A}_1=[{\bf a}_{1,1},\cdots,{\bf a}_{1,d_1}]$ onto $\Span({\bf A}_2)$ is denoted by ${\bf B}_1 = [{\bf b}_{1,1},\cdots,{\bf b}_{1,d_1}]$, i.e.,
$
  {\bf b}_{1,k} = \Proj_{\Span({\bf A}_2)}({\bf a}_{1,k}).
$
If ${\bf A}_1$ and ${\bf A}_2$ are independent, we have
$
  {\bf B}_1 = {\bf V}_2{\bm\Omega},
$
where ${\bm \Omega}\in\mathbb{R}^{d_2\times d_1}$ is a Gaussian random matrix.

This can be readily verified by using the fact that
$
  {\bf B}_1 = {\bf V}_2{\bf V}_2^{\rm T}{\bf A}_1 =: {\bf V}_2{\bm\Omega},
$
where ${\bm\Omega} = (\omega_{i,j}) := {\bf V}_2^{\rm T}{\bf A}_1$.
Because ${\bf A}_1$ is independent with $\Span({\bf A}_2)$, as well as its orthonormal basis ${\bf V}_2$,
the distribution of ${\bf A}_1$ is not influenced, if we first condition ${\bf V}_2$ and regard it as a given matrix.
Consequently, we can readily check that $\omega_{i,j}$ are \emph{i.i.d.} zero mean Gaussian random variables.
\end{rem}

Recalling that ${\bf U}_0^{\rm T}{\bf U}_2={\bf 0}$, which means ${\bf u}_{0,i}^{\rm T}{\bf u}_{2,j}=0$, $1\le i \le d_1$, $1\le j \le d_2$.
According to Lemma  \ref{independence}, we have that ${\bf a}_{0,i}$ and ${\bf a}_{2,j}$ are independent.
Moreover, ${\bf a}_{0,i}$ is independent with ${\set Y}_2$ and thus independent with its orthogonal complement, ${\set Y}_2^{\perp}$.
Then according to Remark \ref{Gaussian-projection}, the projection of ${\bf A}_{0,1:k}\in\mathbb{R}^{n\times k}$ onto ${\set Y}_2^{\perp}$ can be written as
\begin{align}\label{A0-OMG}
{\bf A}_{0,1:k}^{\perp}:=\left[{\bf a}_{0,1}^{\perp}, \cdots, {\bf a}_{0,k}^{\perp}\right]={\bf V}_2^{\perp}{\bm\Omega}_{1:k},
\end{align}
where ${\bf a}_{0,j}^{\perp}:=\Proj_{{\set Y}_2^{\perp}}({\bf a}_{0,j})$,
${\bf V}_2^{\perp}\in\mathbb{R}^{n\times (n-d_2)}$ is an arbitrary orthonormal basis of ${\set Y}_2^\perp$,
and ${\bm\Omega}_{1:k}\in\mathbb{R}^{(n-d_2)\times k}$ is a Gaussian random matrix.
According to the orthonormality of ${\bf V}_2^{\perp}$, we normalize both sides of \eqref{A0-OMG} as
\begin{align}\label{barA0-OMG-pre}
\bar{\bf A}_{0,1:k}^{\perp}={\bf V}_2^{\perp}\bar{\bm\Omega}_{1:k},
\end{align}
where $\bar{{\bm\Omega}}_{1:k}$ denotes the column-normalized ${\bm\Omega}_{1:k}$.
Because left multiplying an orthonormal matrix does not change its singular value, we have
\begin{equation}\label{barA0-OMG}
s_{\max}\left(\bar{\bf A}_{0,1:k}^{\perp}\right)=s_{\max}\left(\bar{\bm\Omega}_{1:k}\right).
\end{equation}
Combining \eqref{A1-A0} and \eqref{barA0-OMG}, and using Remark \ref{P4-nd}, we have
\begin{align}\label{eq-bound-numerator}
\mathbb{P}\left(s_{\max}^2\left(\bar{\bf A}_{1,1:k}^\perp\right)<1+\varepsilon_3\right)
&= \mathbb{P}\left(s_{\max}^2\left(\bar{\bm\Omega}_{1:k}\right)<1+\varepsilon_3\right)\nonumber\\
&>1-{\rm e}^{-c_{2,3}(\varepsilon_3)n}
\end{align}
hold when $n>c_{1,3}(\varepsilon_3)d_2$.

Plugging both bounds of denominator and numerator, i.e., \eqref{eq-bound-denominator} and \eqref{eq-bound-numerator}, into \eqref{eq-1-betak2-bound1}, we can get
  \begin{align}
  1-\beta_k^2\le & \left(1-\lambda_k^2\right)\left(1+\varepsilon_1\right)\frac{1+\varepsilon_3}{1-\varepsilon_2}\nonumber\\
  = & \left(1-\lambda_k^2\right)
  \left(1+\frac{\varepsilon_2+\varepsilon_3}{1-\varepsilon_2}+\varepsilon_1\frac{1+\varepsilon_3}{1-\varepsilon_2}\right),
  \end{align}
with probability at least $1-\sum_{l=1}^3{\rm e}^{-c_{2,l}(\varepsilon_l)n}$ for any $n>\max\{c_{1,l}(\varepsilon_l)\}d_2$.
Let $\varepsilon_2\le1/2$, $\varepsilon_3\le1/2$,
$\varepsilon_1=\varepsilon_2=\varepsilon_3=:\varepsilon/7$, we have
$
\frac{\varepsilon_2+\varepsilon_3}{1-\varepsilon_2}+\varepsilon_1\frac{1+\varepsilon_3}{1-\varepsilon_2} \le 2\left(\varepsilon_2+\varepsilon_3\right)+3\varepsilon_1=\varepsilon.
$
By using Remark \ref{exp-summation} to reshape the probability,
we readily complete the proof.

\subsection{Proof of Lemma \ref{lem-property7}}
\label{proof-lem-property7}

We will calculate the inner product between $\bar{\bf a}_{1,k}^{\perp}$ and ${\bf b}_k^\perp$.
Recalling that ${\bf b}_k$ is the projection of ${\bf a}_{1,k}$ onto $\Span({\bf A}_{1,1:k-1})$, it is not obvious whether $\bar{\bf a}_{1,k}^{\perp}$ and ${\bf b}_k^\perp$ are independent, and this aggravate the problem to estimate their inner product directly.
In order to solve this, therefore, we have to find the relationship between product and projection and then convert the problem to the situation described in Remark \ref{rem-p2epsilon1new}.

Recalling the previous result that $\bar{\bf a}_{1,k}^{\perp}=\bar{\bf a}_{0,k}^{\perp}$ in \eqref{A1-A0} and using the fact of ${\bf b}_k^{\perp}\in \Span(\bar{\bf A}_{1,1:k-1}^{\perp})$, we write
\begin{align}
\left|\langle \bar{\bf a}_{1,k}^{\perp}, {\bf b}_k^{\perp} \rangle\right| &= \left|\langle \bar{\bf a}_{0,k}^{\perp}, {\bf b}_k^{\perp} \rangle\right|\nonumber\\
&\le \left\|\Proj_{\Span(\bar{\bf A}_{1,1:k-1}^{\perp})}(\bar{\bf a}_{0,k}^{\perp})\right\|\nonumber\\
&= \left\|\Proj_{\Span(\bar{\bf A}_{0,1:k-1}^{\perp})}(\bar{\bf a}_{0,k}^{\perp})\right\|.\label{eq5above}
\end{align}

Now we need to construct an orthonormal basis for $\Span(\bar{\bf A}_{0,1:k-1}^{\perp})$ and build its connection with $\bar{\bf a}_{0,k}^{\perp}$.
Recalling Remark \ref{Gaussian-projection} and the deduction in the proof of Lemma \ref{lem-property8}, we reshape \eqref{barA0-OMG-pre} as
$$
\left[\bar{\bf A}_{0,1:k-1}^{\perp}, \bar{\bf a}_{0,k}\right] = {\bf V}_2^{\perp} \left[\bar{\bm\Omega}_{1:k-1}, \bar{\bm{\omega}}_k\right],
$$
where $\bar{\bm{\omega}}_k$ denotes the last column of $\bar{\bm\Omega}_{1:k}\in\mathbb{R}^{(n-d_2)\times k}$.
We next apply Gram-Schimidt orthogonalization to ${\bm\Omega}_{1:k-1}$ and get ${\bf W}_{1:k-1}$, which is an orthonormal basis for $\Span(\bar{\bm\Omega}_{1:k-1})$.
Because of the orthonormality of ${\bf V}_2^{\perp}$, ${\bf V}_2^{\perp}{\bf W}_{1:k-1}$ is an orthonormal basis for $\Span(\bar{\bf A}_{0,1:k-1}^{\perp})$.
As a consequence, we are able to calculate the RHS of \eqref{eq5above} as
\begin{align}
\left\|\Proj_{\Span(\bar{\bf A}_{0,1:k-1}^{\perp})}(\bar{\bf a}_{0,k}^{\perp})\right\|&=\left\|\left({\bf V}_2^{\perp}{\bf W}_{1:k-1}\right)^{\rm T}\bar{\bf a}_{0,k}^{\perp}\right\|\nonumber\\
&= \left\|\left({\bf V}_2^{\perp}{\bf W}_{1:k-1}\right)^{\rm T} {\bf V}_2^{\perp}\bar{\bm{\omega}}_k  \right\|\nonumber\\
&=\left\|{\bf W}_{1:k-1}^{\rm T}\bar{\bm{\omega}}_k\right\| \label{gaussian}.
\end{align}
Recalling that $\bar{\bm\Omega}_{1:k}$ is a column-normalized Gaussian random matrix, $\bar{\bm{\omega}}_k$ should be independent with each column of $\bar{\bm\Omega}_{1:k-1}$, and thus independent with $\Span(\bar{\bm\Omega}_{1:k-1})=\Span({\bf W}_{1:k-1})$.
Combining \eqref{eq5above} and  \eqref{gaussian}, and using Remark \ref{rem-p2epsilon1new}, we have
\begin{align}
\mathbb{P}\left(\left|\langle \bar{\bf a}_{1,k}^{\perp}, {\bf b}_k^{\perp} \rangle\right|^2 > \varepsilon_4\right)
\le \mathbb{P}\left(\left\|{\bf W}_{1:k-1}^{\rm T}\bar{\bm{\omega}}_k\right\|^2>\varepsilon_4\right) < {\rm e}^{-c_{2,4}(\varepsilon_4)n}
\end{align}
for all $n\ge c_{1,4}d_2$.
The proof is completed.

\bibliographystyle{IEEEtran}
\bibliography{mybibfile}

\end{document}